\documentclass[journal]{IEEEtran}
\ifCLASSINFOpdf
\else
\fi
\usepackage{cite}
\usepackage{amsmath,graphicx,amssymb,color,textcomp,amsfonts,subfigure,enumerate,bm}
\usepackage{extarrows}
\usepackage{makecell,multirow,diagbox,stfloats}
\newcommand{\reffig}[1]{Fig. \ref{#1}}
\newcommand{\refeq}[1]{(\ref{#1})}
\renewcommand\arraystretch{1.5}
\usepackage{amsthm}
\usepackage{algorithmicx,algorithm,balance}
\usepackage{algpseudocode}
\usepackage{colortbl}
\usepackage{array}
\usepackage{setspace}
\usepackage{booktabs}
\usepackage{cuted}
\usepackage{forloop}
\usepackage{bbm}
\newcounter{ct}

\newtheorem{proposition}{Proposition}

\begin{document}
	\title{{
	Resilient UAV Swarm Communications with Graph Convolutional Neural Network
	}}
	\author{
		Zhiyu Mou, Feifei Gao, Jun Liu, and Qihui Wu
		
			\thanks{Manuscript received June 15, 2021; revised September 10, 2021; accepted October
			 18, 2021. 
			 This work was supported in part by National Key Research and Development Program of China (2018AAA0102401), by the National Natural Science Foundation of China under Grant 61831013, by Beijing Municipal Natural Science Foundation
under Grant \{L182042, 4212002\}, and by Tsinghua University-China Mobile Communications Group Co.,Ltd. Joint Institute.
This work was also supported in part by the National Natural Science Foundation of China under Grant 61902214.
			The work of Qihui Wu was supported by Project of Major Scientific Instrument, Natural Science Foundation of China (NSFC) under Grant 61827801.
			 \emph{(Corresponding author: Jun Liu, Feifei Gao.)}}
		
	\thanks{
 		{Z. Mou, F. Gao are with the Institute for Artificial Intelligence, Tsinghua University (THUAI), State Key Lab of Intelligence Technologies and Systems, Tsinghua University, Beijing National Research Center for Information Science and Technology (BNRist), Department of Automation, School of Information Science and Technology, Tsinghua University, Beijing 100084, China (email: mouzy20@mails.tsinghua.edu.cn, feifeigao@ieee.org). }}
 		\thanks{{J. Liu is with the Institute of Network Sciences and Cyberspace, Tsinghua University, Beijing 100084, China, and also with the Beijing National Research Center for Information Science and Technology (BNRist), Tsinghua University, Beijing 100084, China (email: juneliu@tsinghua.edu.cn).}}
 		\thanks{{Q. Wu is with the College of Electronic and Information Engineering, Nanjing University of Aeronautics and Astronautics, Nanjing 210016, China (email: wuqihui2014@sina.com).}}
	}
	
	\maketitle	
	\date{}
\pagestyle{empty}  
\thispagestyle{empty} 
\begin{abstract}  
		In this paper, we study the self-healing problem of unmanned aerial vehicle (UAV) swarm network (USNET) that is required to quickly rebuild the communication connectivity under unpredictable external {destructions} (UEDs). Firstly, to cope with the \emph{one-off UEDs}, we propose a graph convolutional neural network (GCN) that can find the recovery topology of the USNET in an on-line manner. Secondly, to cope with \emph{general UEDs}, we develop a GCN based trajectory planning algorithm that can make UAVs rebuild the communication connectivity during the self-healing process. We also design a meta learning scheme to facilitate the on-line executions of the GCN. Numerical results show that the proposed algorithms can rebuild the communication connectivity of the USNET more quickly than the existing algorithms under both one-off UEDs and general UEDs.
		The simulation results also show that the meta learning scheme can not only enhance the performance of the GCN but also reduce the time complexity of the on-line executions.
	\end{abstract}
\section{Introduction}
Unmanned aerial vehicle (UAV) swarm network (USNET) that contains hundreds or even thousands UAVs usually works in open, sometimes even harsh environments and is susceptible to external disruptions \cite{book}. Since the failure of any part of UAVs could be a fatal blow to
the entire USNET, the resilient USNETs with the self-healing capacity are urgently demanded in various applications, such as data collections \cite{yuzhang,zhiyumou}, rescue \cite{search_and_rescue}, security and surveillance \cite{security, jianweizhao}, etc. Researchers have studied the self-healing mechanisms for USNETs in multiple tasks.
For example, the authors of \cite{re_1} developed a real-time resilient method based on the communication connectivity of multi-UAV systems. The authors of \cite{re_2} proposed an intrusion detection scheme based on data exchanging through communication connections to improve the security resilience of the UAV network. Moreover, the authors of \cite{re_3} developed the resilient algorithms for localization, gathering, and network configurations that highly depend on the communication connectivity of the USNET. Obviously, the \emph{communication connectivity} plays an important role in different kinds of self-healing mechanisms, and thus the  \emph{self-healing of the communication connectivity} (SCC) becomes a basic requirement for various resilient USNETs.  

Many algorithms have been developed to deal with the SCC problem for the wireless sensor networks (WSNs) \cite{MCA,LeDiR,AuR,DARA,DORMs,RIM,PCR,HLNF,csds,hero,net}, and were later extended to the USNETs  \cite{neural_model,sidr}. However, there still remain several challenges to the SCC problem in USNET. Firstly, many existing algorithms \cite{LeDiR,AuR,DARA,RIM,PCR,net} are heuristic  and may not be able to guarantee the communication connectivity of the USNET. For example, these algorithms could not work when the number of UAVs is large, especially under massive {destructions}. Other algorithms \cite{DORMs,HLNF,csds,hero,sidr,neural_model} could make sure that the UAVs rebuild the communication connectivity but at the cost of lots of resources, such as self-healing time and communication overheads. 
The second challenge lies in the high time complexities during real-time executions. It is worth noting that the real-time execution time complexity is an important indicator to evaluate the resilience of the USNET, since it relates to the self-healing time and even the degree of {destructions} \cite{survey}. For example, the algorithm in \cite{MCA} needs to find the global cut vertexes of the WSN during the self-healing process, which makes its on-line execution time complexity increase with the size of the WSN. The algorithm in \cite{csds} needs to calculate the optimal critical sensors for WSNs during on-line executions, which may consume a lot of time. 

The third challenge is the difficulty in dealing with complex destructions. The external {destructions} can be divided into \emph{predictable external {destructions}} (PEDs) and \emph{unpredictable external {destructions}} (UEDs). PEDs can be mitigated or even avoided  by finding the pattern of {destructions}, while UEDs could have serious impacts on the USNETs and should be carefully handled \cite{book}. UEDs can be further divided into \emph{one-off UEDs} and \emph{general UEDs}. One-off UEDs happen only once and can destruct a random number of UAVs simultaneously. Almost all the existing UED algorithms \cite{MCA,LeDiR,AuR,DARA,DORMs,RIM,PCR,HLNF,csds,hero,net, neural_model,sidr} are proposed for one-off UEDs. Moreover, many of the UED algorithms
\cite{MCA,LeDiR,DARA,RIM,PCR} were designed regarding to the failure of only one UAV {in one-off UEDs}, which is relatively basic and simple. Other UED algorithms \cite{AuR, HLNF,DORMs} were developed for the failure of multiple UAVs {in one-off UEDs}, but exclusively focused on the scenarios where a small number of UAVs were destructed. In fact, a general UED\footnote{{A general UED can also be regarded as a sequence of one-off UEDs happened at different time steps.}} can destruct any number of UAVs at random time steps, which is more common in practice but more difficult to handle. However, to the best of our knowledge, the general UEDs have not been considered in literatures, yet.

In this paper, we study the SCC problem of the USNET under two types of UEDs, separately. To cope with one-off UEDs, we propose a graph convolutional operation (GCO) that can theoretically guarantee the SCC of the USNET. We then extend the GCO to a graph convolutional neural network (GCN) to minimize the SCC time of the USNET. Moreover, we design a meta learning scheme for the GCN to reduce the time complexity of on-line executions. To cope with general UEDs, we develop a monitoring mechanism that can detect UEDs for UAVs and design a self-healing trajectory planning algorithm based on the GCN and the monitoring mechanism. The numerical results show that the proposed algorithms can rebuild the communication connectivity of the USNET much faster than the existing algorithms under both one-off and general UEDs. The simulation results also show that the meta learning scheme can make the GCN converge faster and reduce the time of on-line executions under both types of UEDs.

\begin{table*}[t]
	\centering
	\caption{{The Summarization of Abbreviations}}
	\label{table:abbreviations}
	\begin{tabular}{cc|cc}
		\specialrule{0em}{1pt}{1pt}
		\hline
		\rowcolor[gray]{0.9}
		\small\textbf{{Abbreviations}}&\small \textbf{{Full Name}}&	\small\textbf{{Abbreviations}}&\small \textbf{{Full Name}}\\
		\hline
		\small\makecell[c]{{UAV}}&\small\makecell[c]{{unmanned aerial vehicle}}&\small\makecell[c]{{USNET}}&\small\makecell[c]{{unmanned aerial vehicle}\\ {swarm network}}\\
		\hline
		\small\makecell[c]{{SCC}}&\small\makecell[c]{{self-healing of the} \\{communication connectivity}}&\small\makecell[c]{{WSN}}&\small\makecell[c]{{wireless sensor network}}\\
		\hline
		\small\makecell[c]{{PED}}&\small\makecell[c]{{predictable external {destruction}}}&\small\makecell[c]{{UED}}&\small\makecell[c]{{unpredictable external {destruction}}}\\
		\hline
		\small\makecell[c]{{GCO}}&\small\makecell[c]{{graph convolutional operation}}&\small\makecell[c]{{GCN}}&\small\makecell[c]{{graph convolutional network}}\\
		\hline
		\small\makecell[c]{{CCN}}&\small\makecell[c]{{connected communication network}}&\small\makecell[c]{{MCL}}&\small\makecell[c]{{multi-hop of communication link}}\\
		\hline
		\small\makecell[c]{{RUAV}}&\small\makecell[c]{{remaining UAV}}&\small\makecell[c]{{A2A}}&\small\makecell[c]{{air-to-air}}\\
		\hline
		\small\makecell[c]{{CLEC}}&\small\makecell[c]{{communication link establish}\\{condition}}&\small\makecell[c]{{FT}}&\small\makecell[c]{{Fourier transform}}\\
		\hline
		\small\makecell[c]{{CR-MGC}}&\small\makecell[c]{{communication-relaxed meta}\\{graph convolution}\\{(dealing with one-off UEDs)}}&\small\makecell[c]{{VRG}}&\small\makecell[c]{{virtual RUAV graph}}\\
		\hline
		\small\makecell[c]{{GCL}}&\small\makecell[c]{{graph convolutional layer}}&\small\makecell[c]{{mGCN}}&\small\makecell[c]{{meta GCN}}\\
		\hline
		\small\makecell[c]{{IDB}}&\small\makecell[c]{{individual data base}}&\small\makecell[c]{{IISR}}&\small\makecell[c]{{individual index set of RUAVs}}\\
		\hline
		\small\makecell[c]{{CR-MGCM}}&\small\makecell[c]{{communication-relaxed meta}\\ {graph convolution method}\\ {(dealing with general UEDs}\\{using monitoring mechanisms)}}&\small\makecell[c]{{CR-MGCM$_{glob}$}}&\small\makecell[c]{{communication-relaxed meta}\\{graph convolution method}\\{(dealing with general UEDs} \\{using global information)}}\\
		\hline	
	\end{tabular}
\vspace{-0.1cm}
\end{table*}

The rest parts of this paper are organized as follows.  Section \ref{system_model} presents the system models of the SCC problem for USNET. Section \ref{CR-MGC} describes the proposed GCN and meta learning scheme under one-off UEDs. Section \ref{distributed_MGC} focuses on the monitoring mechanisms and trajectory planning algorithm of UAVs under the general UEDs. Simulation results and analysis are provided in Section \ref{section:simulations}, and conclusions are made in Section \ref{section:conclusions}. {The abbreviations are summarized in Table \ref{table:abbreviations}.}
 
\emph{Notations}: $x$, $\mathbf{x}$, $\mathbf{X}$ represent a scalar $x$, a vector $\mathbf x$ and a matrix $\mathbf X$, respectively; $\sum$, $\min$, $\max$ and $\nabla$ denote the sum, minimum, maximum and vector differential operator, respectively; $(x_{ij})$ represents a matrix with element $x_{ij}$ in the $i$-th row and the $j$-th column, and $(\mathbf{X})_{ij}$ represents the element of row $i$ and column $j$ in matrix $\mathbf{X}$; $\left\|\cdot\right\|_2$ and $\left\|\cdot\right\|_\infty$ denote the 2-norm and infinite norm of matrices, respectively; $\cup$, $\cap$ and $\backslash$ represent the union operator, the intersection operator and the difference operator between sets; $|\mathcal{S}|$ represents the number of elements in set $\mathcal{S}$; $\mathbb{R}^{N\times M}$, $\mathbb{S}^N$ and $\mathbb{S}^N_{+}$ represent the $N$-by-$M$ real matrix space, the $N$-by-$N$ symmetric matrix space and the $N$-by-$N$ positive semi-definite matrix space; $\mathbb{N}_{+}$ represents the set of positive integers; $\mathbf{1}_n$ represents an $n$-dimensional vector where the components are all $1$'s; $\mathbbm{1}\{\cdot\}$ represents the indicative function with range $\{0,1\}$, $\leftarrow$ denotes the assignment from right to left, while $\rightarrow$ represents the approximation of the right term by the left term; $\triangleq$ defines the symbol on the left by the equation on the right.
\section{System Model}
\label{system_model}
 \begin{figure*}[t]
	\centering
	\includegraphics[width=180mm]{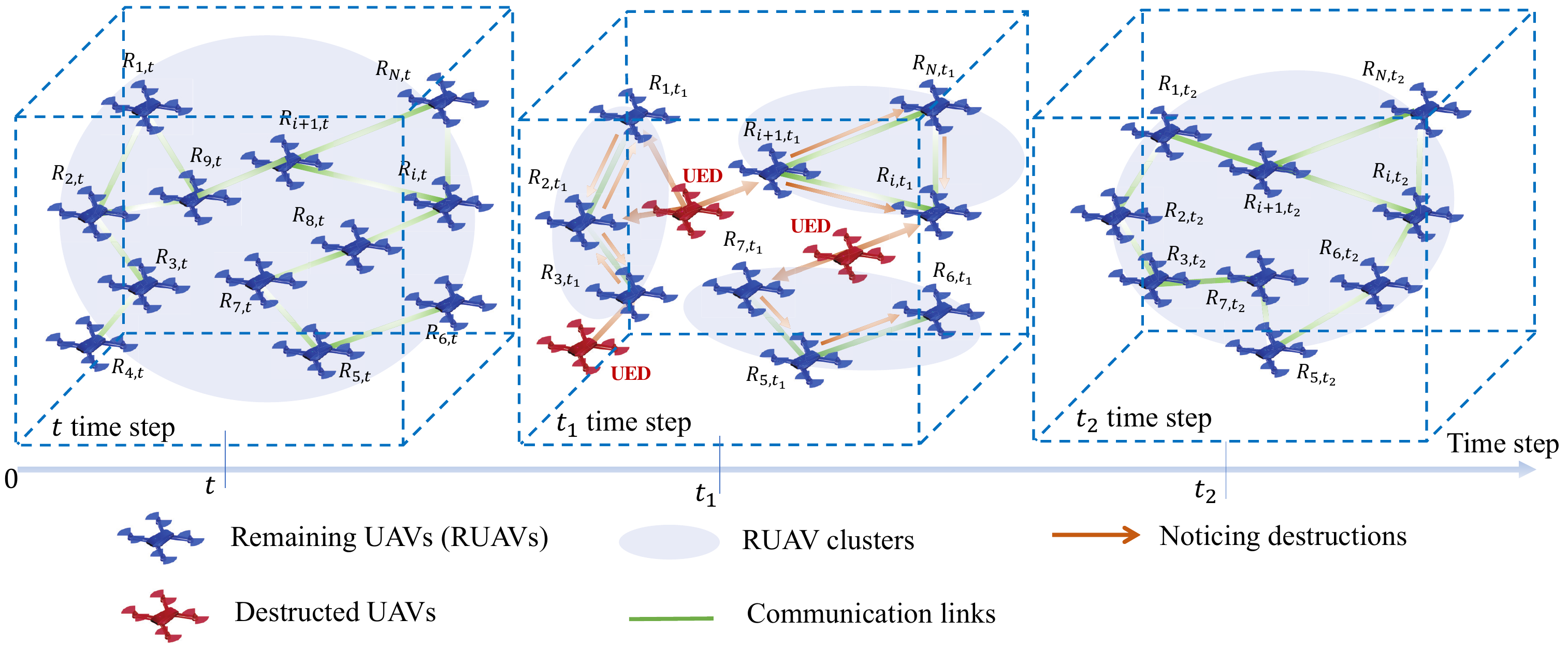}
	\caption{The USNET rebuild its communication connectivity under UEDs.}
	\label{fig:scene}
\end{figure*}
We consider a USNET with $N$ identical UAVs\footnote{The UAVs distribute sparsely to reduce the impacts of the UEDs. Dense gathering can increase the risk of losing more UAVs.}, where each UAV is endowed with a fixed index $i\in\mathcal{N}\triangleq\{1,2,...,N\}$, 
as shown in \reffig{fig:scene}. 
Establish an $X$-$Y$-$Z$ Cartesian coordinate for the USNET, and let the position of the $i$-th UAV at time step $t$ be $\mathbf{p}_{i,t}=[x_{i,t},y_{i,t},z_{i,t}]^T,\;i\in\mathcal{N}$, where $x_{i,t}$, $y_{i,t}$ and $z_{i,t}$ represent the $X$, $Y$ and $Z$ axis components, respectively. Each UAV can transmit signals to other UAVs with constant power $P$. The $i$-th and the $i'$-th UAVs can establish a communication link $e_{ii',t}$ (or $e_{i'i,t}$) at time step $t$ when the powers of the signals received by the $i$-th and the $i'$-th UAV from each other, denoted as $P_b(i,i')$ and $P_b(i',i)$, both exceed a threshold $P_0$, i.e.,
\begin{align}
\label{cl_1}
P_b(i,i')\ge P_0\;\;\text{and}\;\;P_b(i',i)\ge P_0.
\end{align}
Any two UAVs with a valid communication link are marked as \emph{neighbors} of each other. The initial USNET forms a \emph{connected communication network} (CCN), where each UAV can transmit data to any other UAVs in the USNET through \emph{multi-hop of communication links} (MCLs). 

Due to the hash environments, the UEDs can destruct random UAVs at any time step and thus destroy the CCN. The destructed UAVs are forced to detach from the USNET, which can be sensed by their neighbors. The remaining UAVs (RUAVs) react against the destructions and try to restore the CCN by adjusting their positions. Once RUAVs rebuild CCN, they stop flying immediately to avoid gathering denser such that the impact of the next UEDs can be reduced. 
We denote the index set of RUAVs at time step $t$ as $\mathcal{I}_{t}\triangleq\{i|\text{the } i\text{-th UAV remains at time step }t,i\in\mathcal{N}\}$. Let us sort the elements in $\mathcal{I}_t$ in an ascending order, and re-represent it as $\mathcal{I}_t=\{r_1,r_2,...,r_{|\mathcal{I}_t|}\}$, where $r_j$ represents the $j$-th smallest element in $\mathcal{I}_t$, or equivalently, the $j$-th smallest index among all RUAVs, 
$j\in\{1,2,...,|\mathcal{I}_t|\}$. 
Denote the RUAV with index $i$ at time step $t$ as RUAV$_{i,t}$, and then the set of RUAVs at time step $t$ can be defined as $\mathcal{R}_t\triangleq\{\text{RUAV}_{i,t}|i\in\mathcal{I}_t\}=\{\text{RUAV}_{r_1,t},\text{RUAV}_{r_2,t},...,\text{RUAV}_{r_{|\mathcal{I}_t|},t}\}$. 
Assume the magnitude of the flying speed of each UAV is a constant $v_0>0$. The speed of RUAV$_{i,t}$ at time step $t$ can thus be represented as $\mathbf{v}_{i,t}=\check{\mathbf{v}}_{i,t}v$, where $\check{\mathbf{v}}_{i,t}$ is the unit vector of the flying direction, i.e., $\left\|\check{\mathbf{v}}_{i,t}\right\|_2=1$, and $v\in\{v_0,0\}$. 

\subsection{Communication Link Between UAVs}
 We model the communication channels between UAVs as \emph{ air-to-air} (A2A) communication links \cite{A2A, yuzhang_2}. At each time step $t$, the power of the received signals of the $i$-th UAV from the $i'$-th UAV is calculated as\footnote{Note that the units of the variables in \refeq{equation:communication_power} are all dBs.}
\begin{align}
\label{equation:communication_power}
P_b(i,i')=P+G_1+G_2-L(\mathbf{p}_{i,t},\mathbf{p}_{i',t})-p_{\xi}(\mathbf{p}_{i,t},\mathbf{p}_{i',t}),
\end{align}
where $G_1$ and $G_2$ represent the constant antenna gains of the receiving and transmitting UAVs, respectively,  $L(\mathbf{p}_{i,t},\mathbf{p}_{i',t})$ is the large-scale fading effect, and $p_{\xi}(\mathbf{p}_{i,t},\mathbf{p}_{i',t})$ is the small-scale fading effect. Since there is no ground obstacle for USNET, the large-scale effect $L(\mathbf{p}_{i,t},\mathbf{p}_{i',t})$ can be expressed as 
\begin{align}
\label{cl_3}
L(\mathbf{p}_{i,t},\mathbf{p}_{i',t})=10\alpha\log_{10}\bigg(\frac{4\pi \left\|\mathbf{p}_{i,t}-\mathbf{p}_{i',t}\right\|_2f_c}{v_c}\bigg),
\end{align}
where $\alpha>0$ is the path loss exponent, $f_c$ is the electromagnetic wave frequency, and $v_c$ is the speed of light. The small-scale fading effect $p_{\xi}(\mathbf{p}_{i,t},\mathbf{p}_{i',t})$ is usually modeled as the Rice function \cite{rice}, i.e.,
\begin{align}
\label{cl_4}
p_{\xi}(\mathbf{p}_{i,t},\mathbf{p}_{i',t})&=\frac{\left\|\mathbf{p}_{i,t}-\mathbf{p}_{i',t}\right\|_2}{\sigma_0^2}\exp\bigg(\frac{-\left\|\mathbf{p}_{i,t}-\mathbf{p}_{i',t}\right\|_2^2-\rho^2}{2\sigma_0^2}\bigg)\notag\\&I_0(2K\left\|\mathbf{p}_{i,t}-\mathbf{p}_{i',t}\right\|_2),
\end{align}
where $\rho$ and $\sigma_0$ represent the strength of the dominant
and scattered (non-dominant) paths, respectively, $I_0$ is the $0$-th order modified Bessel function of the first kind, and $K=\frac{\rho^2}{2\sigma_0^2}$ is the Rice factor. Since the received signal power $P_b(i,i')$ only relates to the relative distance between the $i$-th and the $i'$-th UAVs, i.e., $l_{ii',t}= l_{i'i,t}\triangleq\left\|\mathbf{p}_{i,t}-\mathbf{p}_{i',t}\right\|_2$, {the received signal power $P_b(i',i)$ equals to $P_b(i,i')$, i.e.,}
\begin{align}
\label{cl_5}
P_b(i',i)=P_b(i,i').
\end{align}
From \eqref{cl_1}, \eqref{equation:communication_power}, \eqref{cl_3}, \eqref{cl_4} and \eqref{cl_5}, we know that any two distinct UAVs with index $i$ and $i'$ can establish a communication link if their distance $l_{ii',t}$ satisfies:
\begin{align}
\label{CL}
10\alpha&\log_{10}\bigg(\frac{4\pi l_{ii',t}f_c}{v_c}\bigg)+\frac{l_{ii',t}}{\sigma_0^2}\exp\bigg(\frac{-l_{ii',t}^2-\rho^2}{2\sigma_0^2}\bigg)I_0(2Kl_{ii',t})\notag\\&\le P+G_1+G_2-P_0.
\end{align} 
Equation \refeq{CL} is called as the \emph{communication link establish condition} (CLEC). 

\subsection{RUAV Graph}
RUAVs at each time step $t$ can be viewed as an undirected graph $\mathcal{G}_{t}=\{\mathcal{R}_{t}, \mathcal{E}_{t},\mathbf{X}_t\}$ \cite{sidr}, named as \emph{RUAV graph}, where $\mathcal{R}_{t}$ acts as the \emph{node set}, and $\mathcal{E}_{t}=\{e_{ii',t}|i,i'\in\mathcal{I}_t \}$ is the \emph{edge set} containing all the communication links of RUAVs. The third term $\mathbf{X}_t\in\mathbb{R}^{|\mathcal{I}_t|\times 3}$ is the \emph{topology matrix} that concatenates the positions of RUAVs, i.e., $\mathbf{X}_t=[\mathbf{p}_{r_1,t},\mathbf{p}_{r_2,t},...,\mathbf{p}_{r_{|\mathcal{I}_t|},t}]^T$. 
We define an \emph{RUAV cluster} as a subset of $\mathcal{R}_t$, where RUAVs in an RUAV cluster form a local CCN but with no communication links to other RUAV clusters. Denote $C_t\in\mathbb{N}_{+}$ as the number of RUAV clusters at time step $t$. 
Due to the UEDs, the RUAV graph $\mathcal{G}_t$ contains at least one {RUAV cluster} at each time step $t$, i.e., $C_t\ge 1$. 
For example, as shown in \reffig{fig:scene}, the RUAV graph $\mathcal{G}_t$ has $3$ RUAV clusters at time step $t_1$, while emerges to one RUAV cluster and forms a CCN at time step $t_2$.

Define the \emph{adjacency matrix} of the RUAV graph $\mathcal{G}_t$ as $\mathbf{A}_t=(a_{jj',t})\in\mathbb{S}^{|\mathcal{I}_t|}$, where $a_{jj',t}\in\{0,1\}, j,j'\in\{1,2,...,|\mathcal{I}_t|\}$. Note that if $j\ne j'$ and the communication link $e_{r_jr_{j'},t}$ exists between RUAV$_{r_j,t}$ and RUAV$_{r_{j'},t}$, then $a_{jj',t}=a_{j'j,t}=1$; otherwise $a_{jj',t}=a_{j'j,t}=0$. The \emph{degree matrix} of the RUAV graph $\mathcal{G}_t$ is defined as a diagonal matrix $\mathbf{D}_t=\text{diag}(d_{1,t}, d_{2,t},...,d_{{|\mathcal{I}_t|},t})\in\mathbb{S}^{|\mathcal{I}_t|}$,where $d_{j,t}=\sum_{j'=1}^{|\mathcal{I}_t|}a_{jj',t}$ is the number of the neighbors of RUAV$_{r_j,t}$. The \emph{Laplace matrix} of the RUAV graph $\mathcal{G}_t$ is defined as the difference between  $\mathbf{D}_t$ and $\mathbf{A}_t$, i.e.,
\begin{align}
\mathbf{L}_t=\mathbf{D}_t-\mathbf{A}_t.
\end{align}
As the Laplace matrix $\mathbf{L}_t$ is a positive semi-definite matrix \cite{laplace}, we can perform eigenvalue decomposition,
\begin{align}
\label{eigen_factor}
\mathbf{L}_t=\mathbf{U}_t\mathbf{\Lambda}_t\mathbf{U}^T_t,
\end{align}
where $\mathbf{U}_t=[\mathbf{u}_{1,t},\mathbf{u}_{2,t},...,\mathbf{u}_{|\mathcal{I}_t|,t}]$ is a unitary matrix composed of $|\mathcal{I}_t|$ mutually orthogonal eigenvectors, and $\mathbf{\Lambda}_t=\text{diag}(\lambda_{1,t},\lambda_{2,t},...,\lambda_{|\mathcal{I}_t|,t})$  is a diagonal matrix with non-negative eigenvalues.
Notice that 0 must be one of the eigenvalues of $\mathbf{L}_t$, since 
\begin{align}
\renewcommand{\arraystretch}{1}
\mathbf{L}_t\mathbf{1}_{|\mathcal{I}_t|}&=(\mathbf{D}_t-\mathbf{A}_t)\mathbf{1}_{|\mathcal{I}_t|}=
\begin{pmatrix}
d_{1,t}-\sum_{j=1}^{|\mathcal{I}_t|}a_{1j,t} \\
\vdots\\
d_{{|\mathcal{I}_t|},t}-\sum_{j=1}^{|\mathcal{I}_t|}a_{{|\mathcal{I}_t|}j,t}
\end{pmatrix}\notag\\
&=\mathbf{0}=0\mathbf{1}_{|\mathcal{I}_t|},
\end{align}
and $\mathbf{1}_{|\mathcal{I}_t|}$ is one possible corresponding eigenvector. The algebraic multiplicity of the zero eigenvalue $\Omega(\lambda=0|\mathbf{L}_t)$ equals to the number of RUAV clusters $C_t$ of the RUAV set $\mathcal{R}_t$ at each time step $t$, i.e., $\Omega(\lambda=0|\mathbf{L}_t)=C_t$ \cite{laplace}. Hence, if $\Omega(\lambda=0|\mathbf{L}_t)=C_t=1$, then RUAVs form a CCN, while if $\Omega(\lambda=0|\mathbf{L}_t)=C_t>1$ otherwise. 
\subsection{Problem Formulation}
The goal of the SCC problem of resilient USNET is that RUAVs should try to reform CCNs as quickly as possible after UEDs. We first study the SCC problem under one-off UEDs, where the initial USNET is {destructed} by a random UED only once at time step $t$ and self-heal afterwards. For a USNET with $N$ UAVs, there are $2^N$ {case}s of {one-off} UEDs, where different {case}s of {one-off} UEDs {destruct} different number of UAVs with different indexes. Note that not all {case}s of {one-off} UEDs can destroy the communication connectivity of the USNET, and we only consider the {one-off} UEDs that can break up the USNET into more than one RUAV clusters {(see Appendix \ref{appendix:type-one-off-UEDs})}. Denote the flying time of RUAV$_{i,t}$ during the self-healing process as $\phi[i]$. Then the total self-healing time steps can be expressed as $\max_{i\in\mathcal{I}_t}\phi[i]$. Since RUAV$_{i,t}$ should fly in a straight line to reduce $\phi[i]$ and since the {magnitude} of the flying speed is a constant $v_0$, the self-healing time steps $\phi[i]$ is proportional to the flying distance of RUAV$_{i,t}$. Hence, the SCC problem under one-off UEDs is equivalent to finding a topology matrix  $\widetilde{\mathbf{{X}}}_t=[\widetilde{\mathbf{p}}_{r_1,t},\widetilde{\mathbf{p}}_{r_2,t},...,\widetilde{\mathbf{p}}_{r_{|\mathcal{I}_{r,0}|},t}]^T$ that can minimize the largest displacement among all RUAVs, i.e.,
\begin{align}
(\mathbf{P1}):\;\min_{\widetilde{\mathbf{X}}_t}\quad&J_s=\max_{i\in\mathcal{I}_t}v_0\phi[i]=\max_{i\in\mathcal{I}_t}\left\|\widetilde{\mathbf{p}}_{i,t}-\mathbf{p}_{i,t}\right\|_2\label{optimization_single}\\
\operatorname{s.t.}\quad&\widetilde{\mathcal{G}}_t=\{\mathcal{R}_t,\widetilde{\mathcal{E}}_t,\widetilde{\mathbf{X}}_t\}\text{ forms a CCN under CLEC}, \tag{{\ref{optimization_single}{a}}}\label{constraint_1_1_}
\end{align}
where $\mathbf{p}_{i,t}=\mathbf{p}_{i,0}$, and $\widetilde{\mathcal{E}}_t\triangleq\{e_{ii',t}| \widetilde{l}_{ii',t}=\left\|\widetilde{\mathbf{p}}_{i,t}-\widetilde{\mathbf{p}}_{i',t}\right\|_2 \text{satisfies CLEC}, \forall i\neq i', i,i'\in\mathcal{I}_t\}$. 

We next study the SCC problem under the general UEDs, where 
the USNET needs to quickly rebuild its communication connectivity under the general UEDs. {Under the circumstances}, RUAVs can only obtain \emph{partial information} from each other and need to adjust their flying directions continuously during the self-healing process. 
We consider a period of $T$ time steps. Define the \emph{connected time step ratio} $J_c=\frac{1}{T}\sum_{t=1}^T\mathbbm{1}\{C_t=1\}$ as the ratio between the number of time steps when the USNET forms a CCN and the total time steps $T$. Let $J_c$ be the performance indicator of the USNET. Then the SCC problem under the general UEDs can be formulated as a functional optimization problem
\begin{align}
(\mathbf{P2}):\;&\mathop{\rm{max}}_{{\mathbf{v}}_{1,t},{\mathbf{v}}_{2,t},...,{\mathbf{v}}_{N,t}\atop t\in\{1,2,...,T\}}\quad J_c=\frac{1}{T}\sum_{t=1}^T\mathbbm{1}\{C_t=1\}\label{optimization}\\
\operatorname{ s.t.} \;\; &\mathbf{p}_{i,t}=\mathbf{p}_{i,t-1}+{\mathbf{v}}_{i,t}, \;\forall i\in\mathcal{I}_t,t\in\{1,2,...,T\}\tag{{\ref{optimization}{a}}}\label{constraint_1}\\
&\mathcal{I}_T\subseteq\mathcal{I}_{T-1}\subseteq...\subseteq\mathcal{I}_{0},\tag{{\ref{optimization}{c}}}\label{constraint_3}\\
&\refeq{CL},\tag{{\ref{optimization}{d}}}\label{constraint_4}
\end{align}
where \refeq{constraint_1} is the dynamic model of RUAVs, \refeq{constraint_3} represents the general UEDs to the USNET, and \refeq{constraint_4} is the CLEC.

\section{SCC Algorithm for One-off UEDs}
\label{CR-MGC}
Let us consider the SCC problem under one-off UEDs $(\mathbf{P1})$. Inspired from the existing swarm algorithms \cite{fish_swarm,ant_swarm}, one RUAV should pay more attention on the positions of its neighbors during the self-healing process. Since graph neural networks (GNNs) \cite{ICA,deepwalk,gcn_2017, GCN} can efficiently gather the neighbor information for each RUAV, we develop a GNN-based algorithm for $(\mathbf{P1})$.

Analogous to the Fourier transform (FT) in the time domain, we can define the FT of the RUAV graph $\mathcal{G}_t$ by the eigen-decomposition of the Laplace matrix $\mathbf{L}_t$ in \refeq{eigen_factor}, where the eigenvectors $\mathbf{u}_{j,t}$ denote the Fourier modes and the eigenvalues $\lambda_{j,t}$ denote the frequency of the RUAV graph $\mathcal{G}_t$ \cite{GCN}. 
Regarding the topology matrix $\mathbf{X}_t$ as a signal of the RUAV graph $\mathcal{G}_t$, we can define the FT of $\mathbf{X}_t$ as $\breve{\mathbf{X}}_t=\mathbf{U}^T_t\mathbf{X}_t$. Hence, the GCO between $\mathbf{X}_t$ and the convolutional kernel $\mathbf{g}\in\mathbb{R}^{|\mathcal{I}_t|\times 3}$ can be expressed as \cite{GCN_survey}
\begin{align}
\label{gcn_original}
\mathbf{g}\circ\mathbf{X}_t=\mathbf{U}_t[(\mathbf{U}^T_t\mathbf{g})\odot(\mathbf{U}^T_t\mathbf{X}_t)],
\end{align} 
where $\circ$ represents the convolutional operator, and $\odot$ is the Hadamard product.
To decrease the computation complexity of the convolutional kernel, we approximate \refeq{gcn_original} by truncated Chebyshev polynomials of the first class \cite{gcn_2017}, and the GCO can be expressed as
\begin{align}
\mathbf{U}_t[(\mathbf{U}^T_t\mathbf{g})\odot(\mathbf{U}^T_t\mathbf{X}_t)]&=\mathbf{U}_t\bigg(\sum_{s=0}^1\theta_sF_s({\mathbf{\Delta}}_t)\bigg)\mathbf{U}_t\mathbf{X}_t\notag\\
&=\theta_0\mathbf{X}_t+\theta_1(\mathbf{L}_t-\mathbf{I}_{t})\mathbf{X}_t,
\end{align}
where $F_s$ represents the $s$-th term in the Chebyshev polynomials, $\mathbf{I}_{t}\in\mathbb{S}^{|\mathcal{R}_t|}$ is the identity matrix, $\theta_0$ and $\theta_1$ are two constant parameters, and ${\mathbf{\Delta}}_t=\frac{2\mathbf{\Lambda}_t}{\mathbf{\lambda}_{1,t}}-\mathbf{I}_t$. Particularly, we define a hyperparameter $H_t\triangleq-\theta_1>0$ and let $\theta_0-\theta_1=1$. Then we can define a GCO on the RUAV graph $\mathcal{G}_t$ as 
\begin{align}
\label{gcn_equ}
\mathbf{g}\circ\mathbf{X}_t=(\mathbf{I}_t-H_t\mathbf{L}_t)\mathbf{X}_t.
\end{align} 
Based on \refeq{gcn_equ}, we propose a \emph{communication-relaxed meta graph convolution} (CR-MGC) algorithm for the SCC problem under one-off UEDs $(\mathbf{P1})$. The CR-MGC includes the virtual communication relaxing part and the meta graph convolutional network part, as will be stated as follows.
\subsection{Virtual Communication Relaxing}
After UED at time step $t$, the RUAV graph  $\mathcal{G}_t$ cannot form a CCN under the CLEC \refeq{CL}. 
Nevertheless, we here build a \emph{virtual RUAV graph} (VRG), denoted as $\mathcal{G}^v_t=\{\mathcal{R}^v_t,\mathcal{E}^v_t,\mathbf{X}^v_t\}$, that has the same node set and topology matrix with the RUAV graph $\mathcal{G}_t$, but has the different edge set, i.e., $\mathcal{R}^v_t=\mathcal{R}_t$, $\mathbf{X}^v_t=\mathbf{X}_t$, but  $\mathcal{E}^v_t\neq\mathcal{E}_t$. We want the VRG to form a CCN. To this end, we 
design the edge set as $\mathcal{E}^v_t=\{e^v_{ii',t}|l_{ii',t}\le d_t^v,\forall i\ne i',i,i'\in\mathcal{I}_t\}$, where $d_t^v>0$ is a hyperparameter, named as the \emph{virtual distance}. 
This indicates that any two distinct RUAV$_{i,t}$ and RUAV$_{i',t}$ can establish a communication link $e^v_{ii',t}$ in the VRG if their distance is within the range of $d_t^v$.
Since the VRG is expected to form a CCN, the virtual distance $d_t^v$ should be large enough to make RUAVs establish sufficient communication links in the VRG. Obviously, there must exist a minimum threshold $d_{min,t}^v$ that can just guarantee the VRG to form a CCN. We propose 
an algorithm to find such $d_{min,t}^v$ in Algorithm \ref{algorithm:find_d_min}.
\begin{algorithm}[t]
	\normalsize\caption{{Find the Minimum Threshold $d_{min,t}^v$ for the Virtual Distance $d_t^v$} }
	\label{algorithm:find_d_min}
	\setstretch{1} 
	{\bf Inputs:} The topology matrix $\mathbf{X}_t$, the index set of RUAVs $\mathcal{I}_t$.\\
	{\bf Outputs:} The minimum threshold $d_{min,t}^v$.\\
	{\bf Initialize:} An empty set $\mathcal{M}_t$ to store the pair-wise distance.
	\begin{algorithmic}[1]
		\normalsize 
		\State Calculate the distance between each pair of RUAVs and store them in $\mathcal{M}_t$, i.e., $\mathcal{M}_t=\{\left\|\mathbf{p}_{i,t}-\mathbf{p}_{i',t}\right\|_2\;|\;\forall i\ne i', i,i'\in\mathcal{I}_t\}$, the size of $\mathcal{M}_t$ is $|\mathcal{M}_t|=\frac{|\mathcal{I}_t|(|\mathcal{I}_t|-1)}{2}$;
		\State Sort the elements in $\mathcal{M}_t$ in ascending order, i.e., $\mathcal{M}_t=\{m_{\zeta,t}|\zeta\in\{1,2,...,\frac{|\mathcal{I}_t|(|\mathcal{I}_t|-1)}{2}\}\}$, where $m_{1,t}\le m_{2,t}\le ...\le m_{\frac{|\mathcal{I}_t|(|\mathcal{I}_t|-1)}{2},t}$;
		\For{$\zeta=1$ to $\frac{|\mathcal{I}_t|(|\mathcal{I}_t|-1)}{2}$}
		\State $d_{min,t}^v\leftarrow m_{\zeta,t}$;
		\State Calculate the Laplace matrix $\mathbf{L}_{t}$ of the VRG based on $m_{\zeta,t}$; 
		\If{the algebraic multiplicity of the zero eigenvalue of  $\mathbf{L}_{t}$ is 1}
		\State Break;
		\EndIf
		\EndFor
	\end{algorithmic}
\end{algorithm}
In addition, a meaningful $d_t^v$ should not be larger than a maximum threshold $d_{max,t}^v$, by which any two RUAVs in the VRG can establish a communication link. The maximum threshold $d_{max,t}^v$ can be calculated as $d_{max,t}^v=\max_{i,i'\in\mathcal{I}_t}\{\left\|\mathbf{p}_{i,t}-\mathbf{p}_{i',t}\right\|_2\}$.
Hence, we let the virtual distance $d_t^v$ be in the range of $[d_{min,t}^v,d_{max,t}^v]$, or equivalently, we let
\begin{align}
\label{d_range}
d_t^v=\eta d_{min,t}^v+(1-\eta) d_{max,t}^v,
\end{align}
where $\eta\in[0,1]$ is a hyperparameter. Then the VRG can form a CCN. 
The best choice of $\eta$, denoted as $\eta^\star$, will be illustrated in Section \ref{section:hyperparameter}.

We can derive the adjacency matrix of VRG as $\mathbf{A}_t^v=(a_{jj',t}^v)\in\mathbb{S}^{|\mathcal{I}_t|}$, where $a_{jj',t}\in\{0,1\},\forall j,j'\in\{1,2,...,|\mathcal{I}_t|\}$. Note that if $j\ne j'$ and the communication link $e^v_{jj',t}$ exists, then $a_{jj',t}^v=a_{j'j,t}^v=1$, otherwise $a_{jj',t}^v=a_{j'j,t}^v=0$; The degree matrix of VRG is $\mathbf{D}_t^v=\text{diag}(d_{1,t}^v, d_{2,t}^v,...,d_{{|\mathcal{I}_t|},t}^v)\in\mathbb{S}^{|\mathcal{I}_t|}$, where $d_{j,t}^v=\sum_{j'=1}^{|\mathcal{I}_t|}a^v_{jj',t}$; The Laplace matrix of VRG is $\mathbf{L}^v_t=\mathbf{D}^v_t-\mathbf{A}^v_t$.
\subsection{Meta Graph Convolutional Network}
With the Laplace matrix $\mathbf{L}^v_t$ of VRG, we can define a GCO $G(\cdot)$ as
\begin{align}
\label{virtual_gcn_equ}
\mathbf{g}\circ\mathbf{X}_t=G(\mathbf{X}_t)=(\mathbf{I}_t-H_t\mathbf{L}_t^v)\mathbf{X}_t.
\end{align}
We then apply the GCO $G(\cdot)$ to the RUAV graph.

\subsubsection{Theoretical guarantee of GCOs in finding CCNs}
\label{section:gco}
The topology matrix {in the $k$-th} iteration of GCO $G(\cdot)$ is calculated as
\begin{align}
	\mathbf{X}^{k}_t&=(\mathbf{I}_t-H_t\mathbf{L}_t^v)\mathbf{X}^{k-1}_t=...=(\mathbf{I}_t-H_t\mathbf{L}_t^v)^{k-1}\mathbf{X}^{1}_t\notag\\
	&=(\mathbf{I}_t-H_t\mathbf{L}_t^v)^{k}\mathbf{X}_t,
\end{align}
or equivalently
\begin{align}
	\mathbf{X}^{k}_t=G(\mathbf{X}^{k-1}_t)=...=G^{k-1}(\mathbf{X}^{1}_t)=G^{k}(\mathbf{X}_t),
\end{align}
where $k\in\mathbb{N}_+$. {We can prove the following proposition on the GCO $G(\cdot)$. }
\begin{proposition}
Let $\mathbf{c}\in\mathbb{R}^3$ be an arbitrary constant vector. 
	In the metric space $\{\mathbf{X}_t\mid \frac{1}{|\mathcal{I}_t|}\sum_{i\in\mathcal{I}_t}\mathbf{p}_{i,t}=\mathbf{c}\}\subset \mathbb{R}^{|\mathcal{I}_t|\times 3}$, the GCO $G(\cdot)$ is a contraction mapping \cite{cm} when $0< H_t\le\frac{1}{\left\|\mathbf{A}^v_t\right\|_\infty}$. There exists and only exists one topology matrix $\overline{\mathbf{X}}_t\triangleq[\overline{\mathbf{p}}_{r_1,t},\overline{\mathbf{p}}_{r_2,t},...,\overline{\mathbf{p}}_{r_{|\mathcal{I}_t|},t}]^T$ such that
\begin{align}
	\overline{\mathbf{X}}_t=G(\overline{\mathbf{X}}_t)=\lim_{k\rightarrow\infty}G^k(\mathbf{X}_t),
\end{align}
where the positions of RUAV in $\overline{\mathbf{X}}_t$ all have the same value $\mathbf{c}$, i.e., $\overline{\mathbf{X}}_t=[\mathbf{c},\mathbf{c},...,\mathbf{c},]^T$.
\end{proposition}
\begin{proof}
{See Appendix \ref{appendix:proposition_1}.}
\end{proof} 

Therefore, there must exist a $k^*\in\mathbb{N}_+$, at which the obtained topology matrix
\begin{align}
	\label{q*}
	\mathbf{X}^{k^*}_t&=G^{k^*}(\mathbf{X}_t)=(\mathbf{I}_t-H_t\mathbf{L}_t^v)^{k^*}\mathbf{X}_t\notag\\
	&=[\mathbf{p}^{k^*}_{r_1,t},\mathbf{p}^{k^*}_{r_2,t},...,\mathbf{p}^{k^*}_{r_{|\mathcal{I}_t|},t},]^T,
\end{align}
will make the RUAV graph $\mathcal{G}_t$ a CCN under CLEC, where $\mathbf{p}^{k^*}_{r_i,t}$ is the target position for RUAV $R_{r_i,t}$ to move to.

Express $H_t$ as $H_t=\frac{\epsilon}{\left\|\mathbf{A}^v_{t}\right\|_\infty}$, where $\epsilon$ acts as a
hyperparameter with theoretical convergence range $(0,1]$. The best choice of $\epsilon$, denoted as $\epsilon^\star$, is illustrated as follows.
\subsubsection{Choice of $\eta^\star$ and $\epsilon^\star$}
\label{section:hyperparameter}
The performance of the GCO $G(\cdot)$ can be evaluated by two indicators. The first indicator is the number of iterations $k^*$ needed by the GCO $G(\cdot)$ to obtain $\mathbf{X}^{k^*}_t$. The smaller $k^*$ is, the better performance the GCO $G(\cdot)$ will be. The second indicator is the maximum movement distance among all the RUAVs, i.e., $L_{max}=\max_{i\in\mathcal{I}_t}\left\|\mathbf{p}^{k^*}_{i,t}-\mathbf{p}_{i,t}\right\|_2$. The smaller $L_{max}$ is, the better performance the GCO $G(\cdot)$ will be. 
As $\eta$ determines the edge set of VRG $\mathcal{E}^v_t$, $\eta$ will determines $\mathbf{L}^v_t$ and further influence the performance of GCO $G(\cdot)$. In addition, since $\epsilon$ determines $H_t$, $\epsilon$ will also influence the performance of GCO $G(\cdot)$. Hence, we conduct numerical experiments in Section \ref{subsection_simulation:find_best_hyperparameter} to find $\eta^\star$ and $\epsilon^\star$ that can make the GCO $G(\cdot)$ achieve better performance on both indicators $k^*$ and $L_{max}$.


\subsubsection{Backbones of the GCN}
\begin{figure*}
	\centering
	\includegraphics[width=180mm]{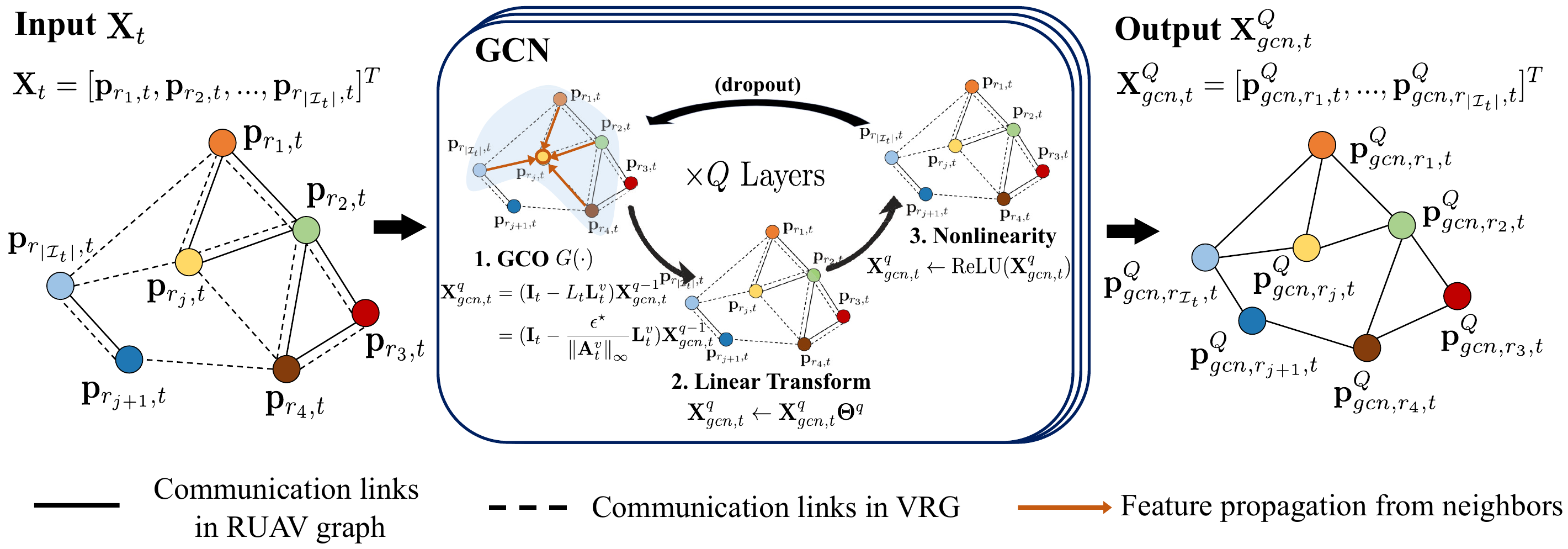}
	\caption{The structure of GCN.}
	\label{fig_GCN}
\end{figure*}

The topology matrix $\mathbf{X}^{k^*}_t$ in \eqref{q*} only satisfies the constraint \refeq{constraint_1_1_}, while does not minimize the objective function ${J}_s$. Therefore, to minimize ${J}_s$, 
we further extend the GCO $G(\cdot)$ to a {graph convolutional network} (GCN). As shown in \reffig{fig_GCN}, the GCN is composed of $Q$ graph convolutional layers (GCLs), where $Q\in\mathbb{N}_+$ is a hyperparameter. 
The $q$-th GCL receives a topology matrix $\mathbf{X}^{q-1}_{gcn,t}$ from the $(q-1)$-th GCL\footnote{The first GCL takes the topology matrix $\mathbf{X}_t$ as the input.} and outputs a topology matrix $\mathbf{X}^q_{gcn,t}$ to the next GCL, $q\in\{1,2,...,Q\}$. Specifically, in the $q$-th GCL, $\mathbf{X}^{q-1}_{gcn,t}$ is processed by the GCO $G(\cdot)$ as
\begin{align}
\mathbf{X}^{q}_{gcn,t}=(\mathbf{I}_t-H_t\mathbf{L}^v_t)\mathbf{X}^{q-1}_{gcn,t}=\big(\mathbf{I}_t-\frac{\epsilon^*}{\left\|\mathbf{A}^v_t\right\|_\infty}\mathbf{L}^v_t\big)\mathbf{X}^{q-1}_{gcn,t}.
\end{align} 
Then the $\mathbf{X}^q_{gcn,t}$ is linearly transformed as
\begin{align}
\mathbf{X}^{q}_{gcn,t}\leftarrow\mathbf{X}^{q}_{gcn,t}\mathbf{\Theta}^q,
\end{align}
where $\mathbf{\Theta}^q$ is the trainable parameter of the $q$-th GCL. In addition, nonlinearities are introduced to the $q$-th GCL by applying the ReLU activation function to $\mathbf{X}^q_{gcn,t}$ as
\begin{align}
\mathbf{X}^{q}_{gcn,t}\leftarrow\text{ReLU}(\mathbf{X}^{q}_{gcn,t}).
\end{align}
Hence, the relationship between $\mathbf{X}^q_{gcn,t}$ and $\mathbf{X}^{q-1}_{gcn,t}$ can be expressed as 
\begin{align}
\mathbf{X}^q_{gcn,t}=\underbrace{\text{ReLU}\bigg(}_{{\text{nonlinearity}}}\underbrace{\big(\mathbf{I}_t-\frac{\epsilon^*}{\left\|\mathbf{A}^v_t\right\|_\infty}\mathbf{L}^v_t\big)}_{{\text{GCO }G(\cdot)}}\underbrace{\mathbf{X}^{q-1}_{gcn,t}\mathbf{\Theta}^q\bigg)}_{{\text{linear transformation}}}.
\end{align}
Note that dropouts\cite{dropout} can be added between two GCLs to increase the generalization ability of the GCN. The output topology matrix  $\mathbf{X}^Q_{gcn,t}=[\mathbf{p}^Q_{gcn,r_1,t},\mathbf{p}^Q_{gcn,r_2,t},...,\mathbf{p}^Q_{gcn,r_{|\mathcal{I}_t|},t}]^T$ of the GCN can form a new RUAV graph $\mathcal{G}^Q_t=\{\mathcal{R}_t, \mathcal{E}^Q_{gcn,t},\mathbf{X}^Q_{gcn,t}\}$, where the edge set  ${\mathcal{E}}_{gcn,t}^Q=\{e_{ii',t}|{l}_{ii',t}=\left\|\mathbf{p}^Q_{gcn,i,t}-\mathbf{p}^Q_{gcn,i',t}\right\|_2 \text{satisfies CLEC}, \forall i\neq i', i,i'\in\mathcal{I}_t\}$. Denote the number of RUAV clusters of the RUAV graph $\mathcal{G}^Q_t$ as $C^Q_t$.
\subsubsection{Loss function design of the GCN}
Denote the loss function of the GCN as $\mathcal{L}(\mathbf{\Theta},\mathbf{X}_t)$, where $\mathbf{\Theta}=\{\mathbf{\Theta}^1,\mathbf{\Theta}^2,...,\mathbf{\Theta}^Q\}$, and $\mathbf{X}_t$ is the input topology matrix to the GCN.
The design of $\mathcal{L}(\mathbf{\Theta},\mathbf{X}_t)$ should be consistent with $(\mathbf{P1})$. 
Specifically, we rewrite $(\mathbf{P1})$ as 
\begin{align}
(\mathbf{P1^\dagger}):\quad\;\min_{\mathbf{X}^Q_{gcn,t}}\quad&{J}_s=\max_{i\in\mathcal{I}_t}\left\|\mathbf{p}^Q_{gcn,i,t}-\mathbf{p}_{i,t}\right\|_2\label{optimization_2}\\
\operatorname{s.t.}\quad&C^Q_t-1\le 0, \tag{{\ref{optimization_2}{a}}}\label{constraint_1_1_1}
\end{align}
where $\widetilde{\mathbf{X}}_t$ in $(\mathbf{P1})$ is substituted by the output of the GCN $\mathbf{X}^Q_{gcn,t}$, and the constraint \refeq{constraint_1_1_} is represented by $C^Q_t-1\le 0$.
Then, we design $\mathcal{L}(\mathbf{\Theta}, \mathbf{X}_t)$ as the Lagrange function of $(\mathbf{P1^\dagger})$ as
\begin{align}
\mathcal{L}(\mathbf{\Theta},\mathbf{X}_t)=\underbrace{\tau(C^Q_t-1)}_{{\text{guarantee the CCN}}}+\underbrace{\max_{i\in\mathcal{I}_t}\left\|\mathbf{p}^Q_{gcn,i,t}-\mathbf{p}_{i,t}\right\|_2}_{{\text{minimize the largest displacement}}},
\end{align}
where the Lagrange multiplier $\tau$ is set as a positive constant. After training the GCN with the designed loss function $\mathcal{L}(\mathbf{\Theta},\mathbf{X}_t)$, the output topology matrix $\mathbf{X}^Q_{gcn,t}$ of the GCN can approximate the solution to $(\mathbf{P1})$, i.e.,  $\mathbf{X}^Q_{gcn,t}\rightarrow\widetilde{\mathbf{X}}_t$.
\subsubsection{Off-line meta learning scheme}
\begin{figure*}
	\centering
	\includegraphics[width=130mm]{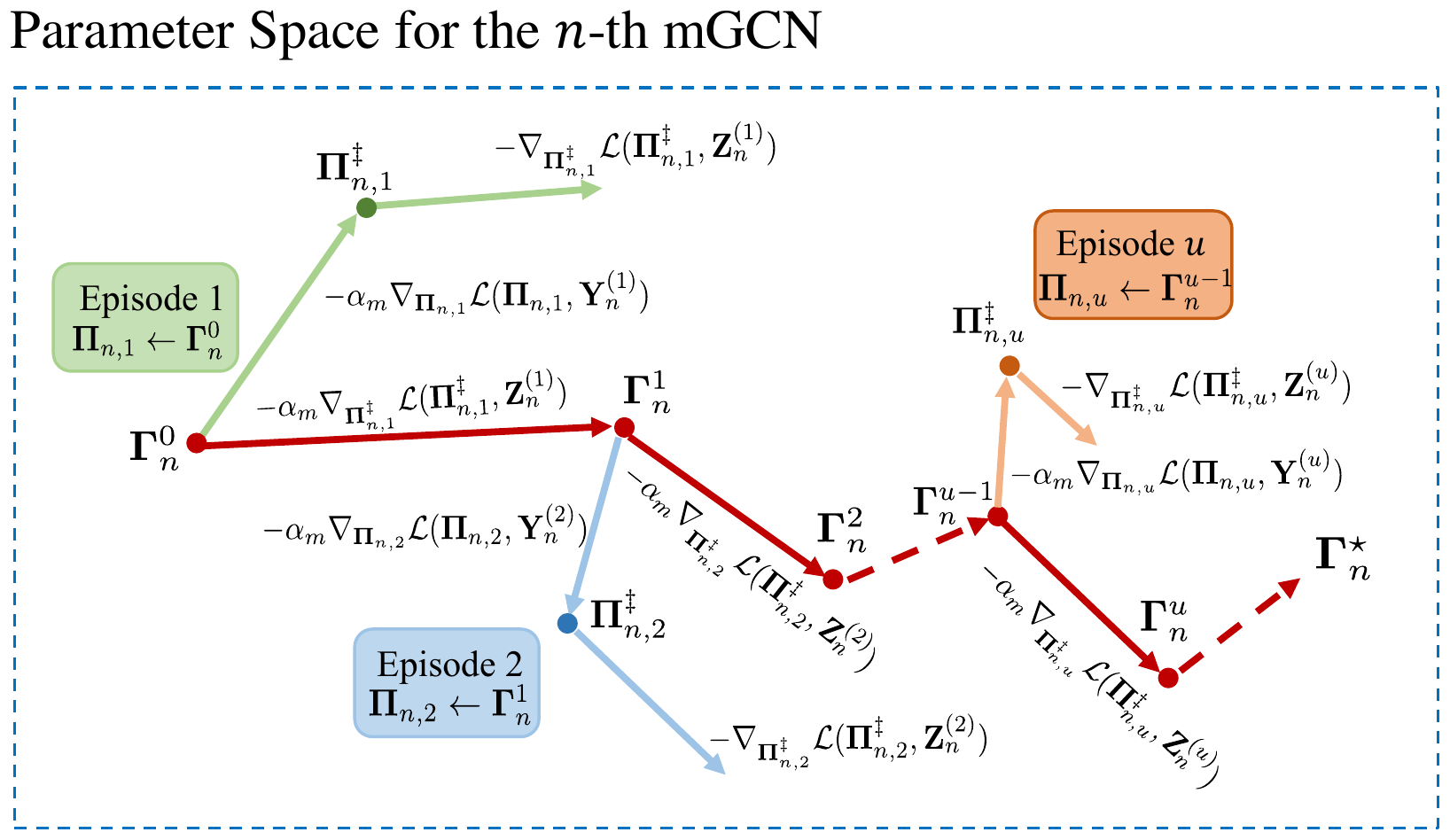}
	\caption{Meta learning procedure for the $n$-th {m}GCN.}
	\label{fig_meta}
\end{figure*}
Notice that different {case}s of UEDs will result in distinct topology matrices $\mathbf{X}_t$, leading to different loss functions $\mathcal{L}(\mathbf{\Theta},\mathbf{X}_t)$ for the GCN. Hence, the GCN should be trained again in an on-line manner when encountering new {case}s of UEDs. 
However, training the GCN from scratch is time consuming and cannot be executed in real-time. Moreover, since there are infinite topology matrices, we cannot train the GCN in advance for each topology matrix. To address these issues, we propose a meta learning scheme for the GCN. The meta learning scheme can find promising initial parameters in an off-line manner to facilitate the on-line trainings \cite{maml}. Specifically, {for a USNET with $N$ UAVs initially, the number $n$ of the RUAVs after one-off UEDs can only be in the range of $\{0,1,2,...,N\}$. We do not need to consider the cases when $n=0$ and $n=1$, since there is either no RUAVs or only one UAV that can form a CCN itself. For the other $N-1$ cases,}
we build $N-1$ GCNs with the same structures as \reffig{fig_GCN}, named \emph{meta} GCNs (mGCNs). {The $n$-th mGCN specifically deals with the case where the number of RUAVs is $n$, $n\in\{2,3,...,N\}$.}
For the $n$-th mGCN, we construct a \emph{support set} $\mathcal{S}_n=\{\mathbf{Y}^{(1)}_{n},\mathbf{Y}^{(2)}_{n},...,\mathbf{Y}^{(U_0)}_{n}\}$ with $U_0$ \emph{support data} $\mathbf{Y}^{(u)}_{n}$, where $U_0\in\mathbb{N}_+$ is the size of $\mathcal{S}_n$, and $\mathbf{Y}^{(u)}_{n}=[\mathbf{p}^{(u,1)}_{n,spt},\mathbf{p}^{(u,2)}_{n,spt},...,\mathbf{p}^{(u,n)}_{n,spt}]^T$ is a randomly generated topology matrix with size $n\times 3$ under the constraint that $\mathbf{Y}^{(u)}_{n}$ cannot make the RUAV graph $\mathcal{G}_t$ form a CCN, $u\in\{1,2,...,U_0\}$.
Meanwhile, we construct a \emph{query set} $\mathcal{W}_n=\{\mathbf{Z}^{(1)}_{n},\mathbf{Z}^{(2)}_{n},...,\mathbf{Z}^{(U_0)}_{n}\}$ with $U_0$ \emph{query data} $\mathbf{Z}^{(u)}_{n}$, where $\mathbf{Z}^{(u)}_{n}=[\mathbf{p}^{(u,1)}_{n,qur},\mathbf{p}^{(u,2)}_{n,qur},...,\mathbf{p}^{(u,n)}_{n,qur}]^T$ is also a randomly generated topology matrix
with size $n\times 3$ under the constraint that $\mathbf{Z}^{(u)}_{n}$ cannot make the RUAV graph $\mathcal{G}_t$ form a CCN. 
We carry out the meta learning in an off-line manner for the $n$-th mGCN, as shown in \reffig{fig_meta}. 
The number of episodes of the meta learning equals to the size $U_0$ of $\mathcal{S}_n$ (or $\mathcal{W}_n$).
In the $u$-th episode, we take $\mathbf{Y}^{(u)}_{n}$ in $\mathcal{S}_n$ and $\mathbf{Z}^{(u)}_{n}$ in $\mathcal{W}_n$ to update the parameter of the $n$-th mGCN at the $u$-th episode $\mathbf{\Gamma}^{u-1}_{n}$. Specifically, a temporary GCN in \reffig{fig_GCN} with parameter ${\mathbf{\Pi}}_{n,u}$ is endowed with $\mathbf{\Gamma}^{u-1}_{n}$, i.e., ${\mathbf{\Pi}}_{n,u}\leftarrow\mathbf{\Gamma}^{u-1}_{n}$. The parameter ${\mathbf{\Pi}}_{n,u}$ is updated in the direction of $\nabla_{{\mathbf{\Pi}}_{n,u}}\mathcal{L}({\mathbf{\Pi}}_{n,u}, \mathbf{Y}^{(u)}_{n})$ by $\alpha_{{meta}}>0$ step size, i.e.,
\begin{align}
\label{support_set_update}
{\mathbf{\Pi}}_{n,u}^\ddagger&={\mathbf{\Pi}}_{n,u}-\alpha_{{meta}}\nabla_{\mathbf{\Pi}_{n,u}}\mathcal{L}(\mathbf{\Pi}_{n,u}, \mathbf{Y}^{(u)}_{n})\notag\\
&=\mathbf{\Gamma}^{u-1}_{n}-\alpha_{{meta}}\nabla_{\mathbf{\Pi}_{n,u}}\bigg[\tau(C^{u,temp}_{n,spt}-1)+\notag\\&
\qquad\quad\max_{n_\beta\in\{1,2,...,n\}}\left\|\mathbf{p}^{(u,n_\beta,temp)}_{n,spt}-\mathbf{p}_{n,spt}^{(u,n_\beta)}\right\|_2\bigg],
\end{align}
where $\mathbf{\Pi}_{n,u}^\ddagger$ is the updated parameter of the temporary GCN, $\mathbf{p}^{(u,n_\beta,temp)}_{n,spt}$ is the $n_\beta$-th row in the output $\mathbf{X}^{(u,temp)}_{n,spt}$ of the temporary GCN, and $C^{u,temp}_{n,spt}$ is the number of RUAV clusters of the RUAV graph $\mathcal{G}_t$ formed by $\mathbf{X}^{(u,temp)}_{n,spt}$. 
The parameter of the $n$-th mGCN is updated in the direction of $\nabla_{\mathbf{\Pi}_{n,u}'}\mathcal{L}(\mathbf{\Pi}_{n,u}^\ddagger,\mathbf{Z}^{(u)}_{n})$ by $\alpha_{{meta}}$ step size, 
i.e.,
\begin{align}
\label{query_set_update}
\mathbf{\Gamma}^u_{n}&=\mathbf{\Gamma}^{u-1}_{n}-\alpha_{{meta}}\nabla_{\mathbf{\Pi}_{n,u}^\ddagger}\mathcal{L}(\mathbf{\Pi}_{n,u}^\ddagger,\mathbf{Z}^{(u)}_{n})\notag\\
&=\mathbf{\Gamma}^{u-1}_{n}-\alpha_{{meta}}\nabla_{\mathbf{\Pi}_{n,u}^\ddagger}\bigg[\tau(C^{u,temp}_{n,qur}-1)+\notag\\
&\qquad\qquad\max_{n_\beta\in\{1,2,...,n\}}\left\|\mathbf{p}^{(u,n_\beta,temp)}_{n,qur}-\mathbf{p}_{n,qur}^{(u,n_\beta)}\right\|_2\bigg],
\end{align}  
where $\mathbf{p}^{(u,j,temp)}_{n,qur}$ is the $u$-th row in the output $\mathbf{X}^{(u,temp)}_{n,qur}$ of the temporary GCN, and $C^{u,temp}_{n,qur}$ is the number of RUAV clusters of the RUAV graph $\mathcal{G}_t$ formed by $\mathbf{X}^{(u,temp)}_{n,qur}$. 
After $U_0$ episodes, we obtain the \emph{meta parameters} of all the $N-1$ mGCNs $\mathbf{\Gamma}^\star_{n}\triangleq\mathbf{\Gamma}^{U_0}_{n}$ that act as the initial parameters for the GCNs during on-line executions.

\subsubsection{On-line executions of the GCN}
When the USNET is {destructed} by one-off UEDs at time step $t$ and the RUAV graph $\mathcal{G}_t$ has $N_0\in\{2,3,...,N\}$ RUAVs, we build the VRG $\mathcal{G}_t^v=\{\mathcal{R}_t^v, \mathcal{E}_t^v,\mathbf{X}_t^v\}$, and calculate the Laplace matrix $\mathbf{L}^v_t$ for the GCN.
Then the GCN will load the meta parameter $\mathbf{\Gamma}_{N_0}^\star$, i.e., $\mathbf{\Theta}\leftarrow\mathbf{\Gamma}_{N_0}^\star$. Next, the GCN will be trained on-line by the gradient descent of the loss function $\mathcal{L}(\mathbf{\Theta}, \mathbf{X}_t)$, i.e.,
\begin{align}
\label{update_rule}
\mathbf{\Theta}\leftarrow\mathbf{\Theta}-\alpha_{{meta}}\nabla_{\Theta}\mathcal{L}(\mathbf{\Theta}, \mathbf{X}_t).
\end{align}
Note that the number of the on-line training episodes, denoted as $M$, is a constant positive integer. After the on-line training, we input $\mathbf{X}_t$ into the GCN, and the GCN outputs the topology matrix $\mathbf{X}^Q_{gcn,t}$ that acts as the solution to $(\mathbf{P1})$, i.e., $\widetilde{\mathbf{X}}_t\leftarrow\mathbf{X}^Q_{gcn,t}$. Each RUAV$_{i,t}$ will fly at a constant speed ${\mathbf{v}}_{i,t}=\frac{v_0}{\left\|\mathbf{p}^{Q}_{gcn,i,t}-\mathbf{p}_{i,t}\right\|_2}(\mathbf{p}^{Q}_{gcn,i,t}-\mathbf{p}_{i,t})$ until reaching point $\mathbf{p}^Q_{gcn,i,t}$.
The process of the CR-MGC algorithm is briefly summarized in Algorithm \ref{algorithm:CR_MGC}.

\begin{algorithm}[t]
	\normalsize\caption{CR-MGC Algorithm (A Brief Process Summary) }
	\label{algorithm:CR_MGC}
	\setstretch{1} 
	{\bf Inputs:} The initial RUAV graph $\mathcal{G}_0=\{\mathcal{R}_0, \mathcal{E}_0, \mathbf{X}_0\}$, and the initial index set of RUAVs $\mathcal{I}_{0}$.\\
	{\bf Outputs:} The solution $\widetilde{\mathbf{X}}_t$ to $(\mathbf{P1})$, the flying trajectories of all RUAVs.\\
	{\bf Initializations:} The parameters of mGCNs $\mathbf{\Gamma}_{n}^0$, the parameter of the GCN $\mathbf{\Theta}$, support sets $\mathcal{S}_n$ and query sets $\mathcal{W}_n$, $n\in\{2,3,...,N\}$. Conduct numerical experiments (shown in Section \ref{subsection_simulation:find_best_hyperparameter}) to determine the $\eta^\star$ and $\epsilon^\star$.\\
	{\bf Off-line Meta Training:}
	\begin{algorithmic}[1]
		\normalsize
		\For{$n=2$ to $N$}
		\For {$u=1$ to $U$} 
		\State Build the VRGs based on $\mathbf{Y}^{(u)}_{n}$ and $\mathbf{Z}^{(u)}_{n}$ separately, and derive the corresponding Laplace matrices. Train one step on parameter $\mathbf{\Gamma}_{n}^{u-1}$ using \refeq{support_set_update}, and update $\mathbf{\Gamma}_{n}^{u-1}$ using \refeq{query_set_update}. 
		\EndFor
		\EndFor
		\State Obtain all the meta parameters $\mathbf{\Gamma}_{n}^\star,n\in\{2,3,...,N\}$.
		
	\end{algorithmic}
	{\bf On-line Executions:}
	\begin{algorithmic}[1]
		\normalsize 
		\State A random one-off UED happens at time step $t$, and the USNET is {destructed} into a RUAV graph $\mathcal{G}_t=\{\mathcal{R}_t, \mathcal{E}_t,\mathbf{X}_t\}$ with $n$ RUAVs.
		\State Build the VRG $\mathcal{G}_t^v=\{\mathcal{R}_t^v, \mathcal{E}_t^v,\mathbf{X}_t^v\}$, and calculate the Laplace matrix $\mathbf{L}^v_t$ for the GCN.
		\State The GCN loads the meta parameter $\mathbf{\Gamma}^\star_{n}$, i.e., $\mathbf{\Theta}\leftarrow\mathbf{\Gamma}^\star_{n}$.
		\State  Train $\mathbf{\Theta}$ with \refeq{update_rule} $M$ episodes, and obtain the output $\mathbf{X}^Q_{gcn,t}$. 
		\State  Let $\widetilde{\mathbf{X}}_t\leftarrow\mathbf{X}^Q_{gcn,t}$. Each RUAV$_{i,t}$ flies at a constant speed ${\mathbf{v}}_{i,t}=\frac{v_0}{\left\|\mathbf{p}^{Q}_{gcn,i,t}-\mathbf{p}_{i,t}\right\|_2}(\mathbf{p}^{Q}_{gcn,i,t}-\mathbf{p}_{i,t})$, $\forall i\in\mathcal{I}_t$ until reach point $\mathbf{p}^Q_{gcn,i,t}$.
	\end{algorithmic}
\end{algorithm}

\section{SCC Algorithm for General UEDs}
\label{distributed_MGC}
In this section, let us consider the SCC problem under the general UEDs $(\mathbf{P2})$. 
To cope with the issue that RUAVs can only obtain partial information, we build an \emph{individual data base} (IDB) model for each UAV and develop a monitoring mechanism that can detect UEDs and the position changing of UAVs. We then propose a self-healing trajectory planning algorithm based on monitoring mechanisms and CR-MGC to cope with the general UEDs.
\subsection{Individual Database Model and Monitoring Mechanisms}
\label{ID_BM}
We embed an IDB $\mathcal{D}_{i,t}=\{\widehat{\mathbf{p}}_{1,t}^i,\widehat{\mathbf{p}}_{2,t}^i,...,\widehat{\mathbf{p}}_{N,t}^i\}\cup \widehat{\mathcal{I}}^i_{t}$ inside the $i$-th UAV that contains two parts, namely the \emph{individual positions} of all UAVs $\{\widehat{\mathbf{p}}_{1,t}^i,\widehat{\mathbf{p}}_{2,t}^i,...,\widehat{\mathbf{p}}_{N,t}^i\}$ and the \emph{individual index set of RUAVs} (IISR) $\widehat{\mathcal{I}}^i_{t}$. The UAVs always know their own positions. Hence, the individual position $\widehat{\mathbf{p}}^i_{i,t}$ in $\mathcal{D}_{i,t}$ of  RUAV$_{i,t}$ equals to the position of  RUAV$_{i,t}$ at each time step $t$, i.e., $\widehat{\mathbf{p}}^i_{i,t}=\mathbf{p}_{i,t}$.
During the self-healing process, the monitoring mechanism is realized through the updating of IDBs.  

\subsubsection{Monitoring the position changing of UAVs by updating the individual positions} 
At each time step $t$, RUAV$_{i,t}$ broadcasts its own position $\mathbf{p}_{i,t}$ to other RUAVs in the same RUAV cluster through MCLs. To better exhibit the SCC algorithm, we ignore the time delay of data transmissions in MCLs, and assume the broadcasting can be completed at time step $t$. If RUAV$_{i,t}$ receives $\mathbf{p}_{i',t}$ at time step $t$, it updates the individual position of the $i'$-th UAV in $D_{i,t}$; otherwise, the old individual position of the $i'$-th UAV in $D_{i,t-1}$ of RUAV$_{i,t}$ does not change, i.e., 
\begin{align}
\label{id_1}
\widehat{\mathbf{p}}^i_{i',t}\leftarrow \left\{
\begin{aligned}
\mathbf{p}_{i',t}, \quad\; &\mbox{if receives } \mathbf{p}_{i',t};\\
\widehat{\mathbf{p}}^i_{i',t-1},\;\; &\mbox{otherwise}.
\end{aligned}
\right.
\end{align} 

\subsubsection{Monitoring the UEDs by updating the IISR} When the $j$-th UAV is destructed at time step $t$, its neighbor RUAV$_{i,t}$ will notice the destruction immediately and 
drop the index $j$ from $\widehat{\mathcal{I}}^i_{t}$, i.e.,
\begin{align}
\label{id_3}
\widehat{\mathcal{I}}^i_{t}\leftarrow\widehat{\mathcal{I}}^i_{t-1}\backslash\{j\}.
\end{align}
RUAVs within the same RUAV cluster share their IISRs through broadcasting, and RUAV$_{i,t}$ updates $\widehat{\mathcal{I}}^i_{t}$ by taking the intersections of all the received IISRs, i.e.,
\begin{align}
\label{id_4}
\widehat{\mathcal{I}}^i_{t}\leftarrow\widehat{\mathcal{I}}^i_{t}\cap\widehat{\mathcal{I}}^{i_1}_{t}\cap\widehat{\mathcal{I}}^{i_2}_{t}\cap...\cap\widehat{\mathcal{I}}^{i_h}_{t}\cap...\cap\widehat{\mathcal{I}}^{i_{|\mathcal{C}_{i,t}|-1}}_{t},
\end{align}
where 
$\mathcal{C}_{i,t}$ represents the RUAV cluster containing RUAV$_{i,t}$, and $\widehat{\mathcal{I}}^{i_h}_t$ represents the received IISR, $i_h\in\mathcal{I}_t,\; h\in\{1,2,...,|\mathcal{C}_{i,t}|-1\}$.
\begin{algorithm}[t]
	\normalsize\caption{CR-MGCM for the $i$-th UAV based on its IDB}
	\label{algorithm:individual_gcn}
	\setstretch{1} 
	{\bf Input:} The IDB $\mathcal{D}_{i,0}=\{\widehat{\mathbf{p}}_{1,0}^i,\widehat{\mathbf{p}}_{2,0}^i\,...,\widehat{\mathbf{p}}_{N,0}^i\}\cup\widehat{\mathcal{I}}^i_{0}$, the $\eta^\star$, $\epsilon^\star$ and $\mathbf{\Gamma}^\star_{n},n\in\{2,3,...,N\}$.\\
	{\bf Outputs:} The speed ${\mathbf{v}}_{i,t}$ of the $i$-th UAV during $t\in\{1,2,...,T\}$.\\
	{\bf Initializations:} An inertia counter $C_I\leftarrow0$, a target position $\mathbf{\Xi}_{i}\in\mathbb{R}^3$, and the inertia $\kappa>0$.
	\begin{algorithmic}[1]
		\normalsize 
		\For{$t=1$ to $T$}
		\State Update IDBs to monitor the UEDs and position changing of UAVs with \refeq{id_1}, \refeq{id_3}, \refeq{id_4}.
		\State Calculate the Laplace matrix $\mathbf{L}_{t}$ of the RUAV graph formed by $\{\widehat{\mathbf{p}}_{1,t}^i,\widehat{\mathbf{p}}_{2,t}^i\,...,\widehat{\mathbf{p}}_{N,t}^i\}$. 
		\If{$\Omega(\lambda=0|\mathbf{L}_{t})>1$}
		\If{$C_I==0$}
		\State Calculate the $d_{min,t}^v$ by Algorithm \ref{algorithm:find_d_min} with inputs $[\widehat{\mathbf{p}}_{r_1,t}^i,\widehat{\mathbf{p}}_{r_2,t}^i\,...,\widehat{\mathbf{p}}_{r_{|\widehat{\mathcal{I}}_{t}^i|},t}^i]^T$ and $\widehat{\mathcal{I}}^i_{t}$, calculate the maximum threshold $d_{max,t}^v$ as  $d_{max,t}^v=\max_{i',i''\in\widehat{\mathcal{I}}_t^i}\{\left\|\widehat{\mathbf{p}}_{i',t}^i-\widehat{\mathbf{p}}_{i'',t}^i\right\|_2\}$, and then $d_{t}^v=\eta^\star d_{min,t}^v+(1-\eta^\star) d_{max,t}^v$;
		\State Build the VRG $\mathcal{G}^v_{t}=\{\mathcal{R}_{t}^v,\mathcal{E}^v_{t},\mathbf{X}^v_{t}\}$, where $\mathcal{R}_{t}^v=\{\text{RUAV}_{i,t}|i\in\widehat{\mathcal{I}}^i_{r,t}\}$, $\mathbf{X}^v_{t}=[\widehat{\mathbf{p}}_{r_1,t}^i,\widehat{\mathbf{p}}_{r_2,t}^i\,...,\widehat{\mathbf{p}}_{r_{|\widehat{\mathcal{I}}_{t}^i|},t}^i]^T$, and $\mathcal{E}^v_{t}=\{e^v_{i'i'',t}|i',i''\in\widehat{\mathcal{I}}_t^i,i'\neq i'', \left\|\widehat{\mathbf{p}}^i_{i',t}-\widehat{\mathbf{p}}^i_{i'',t}\right\|_2\le d_{t}^v \}$. Derive the Laplace matrix $\mathbf{L}^v_{t}$ of the VRG $\mathcal{G}^v_{t}$. 
		\State Load  $\mathbf{\Gamma}^\star_{|\widehat{\mathcal{I}}^i_{t}|}$ to the GCN, i.e., $\mathbf{\Theta}\leftarrow\mathbf{\Gamma}^\star_{|\widehat{\mathcal{I}}^i_{t}|}$, train the GCN $M$ episodes with \refeq{update_rule}.
		\State The GCN outputs $\mathbf{X}^Q_{gcn,t}=[\mathbf{p}^{Q}_{gcn,r_1,t},\mathbf{p}^{Q}_{gcn,r_2,t},...,\mathbf{p}^{Q}_{gcn,r_{|\widehat{\mathcal{I}}^i_{r,t}|},t}]^T$ with input $\mathbf{X}^v_{i,t}$.
		\State Let $\mathbf{\Xi}_{i}\leftarrow\mathbf{p}^{Q}_{gcn,i,t}$, and ${\mathbf{v}}_{i,t}\leftarrow\frac{v_0}{\left\|\mathbf{\Xi}_i-\widehat{\mathbf{p}}^i_{i,t}\right\|_2}(\mathbf{\Xi}_i-\widehat{\mathbf{p}}^i_{i,t})$. Let $C_I\leftarrow C_I+1$.
		\Else
		\State ${\mathbf{v}}_{i,t}\leftarrow\frac{v_0}{\left\|\mathbf{\Xi}_i-\widehat{\mathbf{p}}^i_{i,t}\right\|_2}(\mathbf{\Xi}_i-\widehat{\mathbf{p}}^i_{i,t})$.
		
		\EndIf
		\State Let $C_I\leftarrow0$ if $C_I==\kappa$.
		\Else
		\State ${\mathbf{v}}_{i,t}=\mathbf{0}$, and $C_I\leftarrow0$
		\EndIf
		\If {the $i$-th UAV is destructed}
		\State break
		\EndIf
		\EndFor
	\end{algorithmic}
\end{algorithm}

Define the \emph{global information} $\mathcal{D}_{G,t}$ at time step $t$ as the union of the positions of all UAVs and the index set of RUAVs, i.e., $\mathcal{D}_{G,t}=\{\mathbf{p}_{1,t},\mathbf{p}_{2,t},...,\mathbf{p}_{N,t}\}\cup\mathcal{I}_t$. Note that the monitoring mechanism tries to help RUAVs obtain the latest information about the USNET as mush as possible, but still cannot help all the RUAVs obtain the global information $\mathcal{G}_t$ at each time step $t$. This means that there may exist some certain some time step $t$ at which $\mathcal{D}_{i,t}\neq\mathcal{D}_{G,t}$ for some RUAV$_{i,t}$. Nonetheless, at the time steps when the RUAV graph $\mathcal{G}_t$ forms a CCN, all the RUAVs can obtain the global information $\mathcal{D}_{G,t}$. For example, the USNET forms a CCN at $t=0$, and then there is $\mathcal{D}_{i,0}=\mathcal{D}_{G,0}$.
\subsection{Self-healing Trajectory Planning Algorithm}
Based on the CR-MGC and the monitoring mechanisms, we propose a self-healing trajectory planning algorithm, named CR-MGCM, to cope with the the general UEDs.
The details of CR-MGCM algorithm for each UAV are stated in Algorithm \ref{algorithm:individual_gcn}. In a nutshell, each UAV first loads $\eta^\star$, $\epsilon^\star$, the meta parameters $\mathbf{\Gamma}_{n}^\star$, and the GCN with randomly initialized $\mathbf{\Theta}$. Then during on-line executions, each RUAV monitors the UEDs and position changing of UAVs by updating its IDB. RUAV$_{i,t}$ determines its flying directions by carrying out the on-line execution part of CR-MGC based on the data in $\mathcal{D}_{i,t}$.
Note that for each UAV we set an \emph{inertia} $\kappa\;(\kappa>0)$ that determines the number of time steps to maintain the flying directions before rerunning the on-line execution part of the CR-MGC. The outputs of CR-MGCM of all UAVs act as the solution to $(\mathbf{P2})$.

\subsection{Theoretical Effectiveness of CR-MGCM}
If UAVs always have the global information $\mathcal{D}_{G,t}$, then the CR-MGCM can skip the monitoring mechanism in step ``2" and simply let $\{\widehat{\mathbf{p}}_{1,t}^i,\widehat{\mathbf{p}}_{2,t}^i\,...,\widehat{\mathbf{p}}_{N,t}^i\}\leftarrow\{{\mathbf{p}}_{1,t},{\mathbf{p}}_{2,t}\,...,{\mathbf{p}}_{N,t}\}$ and $\widehat{\mathcal{I}}_{t}^i\leftarrow{\mathcal{I}}_{t}$ for each RUAV$_{i,t}$ in each time step $t$. We refer the CR-MGCM algorithm where UAVs always have the global information $\mathcal{D}_{G,t}$ as CR-MGCM$_{glob}$. Note that CR-MGCM$_{glob}$ is equivalent to the CR-MGC when coping with each single one-off UEDs. Due to the effectiveness of CR-MGC, CR-MGCM$_{glob}$ is effective under one-off UEDs. On the other hand, since the general UEDs can be viewed as the combination of several one-off UEDs at different time steps, the CR-MGCM$_{glob}$ is effective under the general UEDs.

However,
since RUAVs cannot obtain $\mathcal{D}_{G,t}$, they may fly towards wrong directions during the self-healing process, which can make SCC algorithms ineffective. Nonetheless, we prove that CR-MGCM can reach the performance of CR-MGCM$_{glob}$ under the general UEDs.

\begin{proposition}
	\label{propositions_2}
When applying the GCOs $G(\cdot)$ to the topology matrix $\mathbf{X}_t$, the positions of all RUAVs are moving towards their center $\frac{1}{|\mathcal{I}_t|}\sum_{i\in\mathcal{I}_t}\mathbf{p}_{i,t}$.
\end{proposition}
\begin{proof}
{See Appendix \ref{appendix:proposition_2}.}
\end{proof}

Since the GCN is mainly composed of GCO{s} $G(\cdot)$, 
it {tends to} make RUAVs gather towards the center of their positions. However, CR-MGCM makes each RUAV$_{i,t}$ fly towards the \emph{incomplete center} $\frac{1}{|\widehat{\mathcal{I}}_t^i|}\sum_{i'\in\widehat{\mathcal{I}}_{t}^i}\widehat{\mathbf{p}}_{i',t}^i$ that is calculated by the data in $\mathcal{D}_{i,t}$, while the CR-MGCM$_{glob}$ makes each RUAV$_{i,t}$ fly towards the \emph{complete center} $\frac{1}{|\mathcal{I}_t|}\sum_{i'\in\mathcal{I}_t}\mathbf{p}_{i',t}$.
\begin{figure*}
	\centering
	\includegraphics[width=130mm]{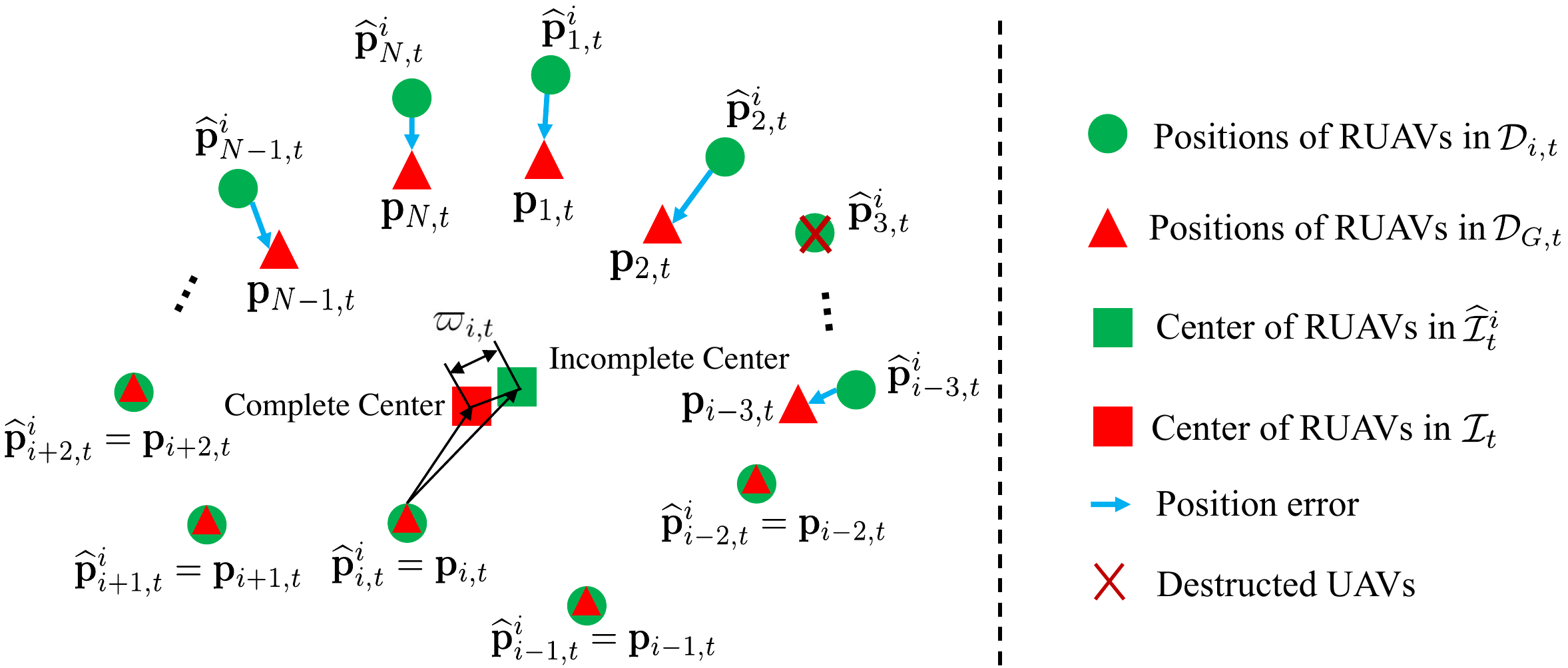}
	\caption{Individual positions in $\mathcal{D}_{i,t}$ and the RUAVs' positions in global information $\mathcal{D}_{G,t}$. }
	\label{fig:relationship_between_centers}
\end{figure*}
\begin{table*}[t]
	\centering
	\caption{Parameter settings of UAVs in the Simulations}
	\label{UAV_param}
	\begin{tabular}{ccc|ccc}
		\specialrule{0em}{1pt}{1pt}
		\hline
		\rowcolor[gray]{0.9}
		\small\textbf{Parameter}&\small \textbf{Values}&\small \textbf{Parameter description}&	\small\textbf{Parameter}&\small \textbf{Values}&\small \textbf{Parameter description}\\
		\hline
		\small\makecell[c]{$P$}&\small\makecell[c]{30 dBm\\ (=1W)}&\small\makecell[c]{Transmitting signal power}&	\small\makecell[c]{$P_0$}&\small\makecell[c]{1.38 dBm\\ (=1.37mW)}&\small\makecell[c]{Receiving signal \\power threshold}\\
		\hline
		\small\makecell[c]{$G_1$,$G_2$}&\small\makecell[c]{6 dBi}&\small\makecell[c]{Antenna gain of\\ receiving and transmitting \\signals}&
		\small\makecell[c]{$\alpha$}&\small\makecell[c]{1} &\small\makecell[c]{$\alpha$ in \refeq{CL}}\\
		\hline
		\small\makecell[c]{$f_c$}&\small\makecell[c]{2.4 GHz} &\small\makecell[c]{Carrier frequency}&
		\small\makecell[c]{$v_c$}&\small\makecell[c]{$3\times 10^8$ m/s} &\small\makecell[c]{Speed of light}\\
		\hline	
		\small\makecell[c]{$\sigma_0^2$}&\small\makecell[c]{5} &\small\makecell[c]{Strength of scattered path}&
		\small\makecell[c]{$K$}&\small\makecell[c]{10} &\small\makecell[c]{Rice factor}\\
		\hline	
		\small\makecell[c]{$v_0$}&\small\makecell[c]{1m/s} &\small\makecell[c]{Magnitude of the\\ speed of UAVs}&
		\small\makecell[c]{{$\alpha_{{meta}}$}}&\small\makecell[c]{{0.01}} &\small\makecell[c]{{learning rate in}\\{the  meta learning}}\\
		\hline	
	\end{tabular}
\end{table*}

We then analyze the difference between the incomplete center and complete center for RUAV$_{i,t}$, as shown in \reffig{fig:relationship_between_centers}. Denote the distance between two centers as $\varpi_{i,t}=\left\|\frac{1}{|\mathcal{I}_t|}\sum_{i'\in\mathcal{I}_t}\mathbf{p}_{i',t}-\frac{1}{|\widehat{\mathcal{I}}_t^i|}\sum_{i'\in\widehat{\mathcal{I}}_t^i}\widehat{\mathbf{p}}_{i',t}^i\right\|_2$, which can be expanded as \eqref{equ:center_1}.
Notice that $\mathcal{I}_t\subseteq\widehat{\mathcal{I}}^i_{t}$ always holds for $\forall i\in\mathcal{I}_t$ and $\forall t\in\{1,2,...,T\}$, since $\widehat{\mathcal{I}}^i_{t}$ has no chance to drop the elements in $\mathcal{I}_t$. Hence, there is $\mathcal{I}_t\backslash\widehat{\mathcal{I}}^i_{t}=\varnothing$, which indicates  $\sum_{i'\in\mathcal{I}_t\backslash\widehat{\mathcal{I}}^i_{t}}\mathbf{p}_{i',t}=\sum_{i'\in\varnothing}\mathbf{p}_{i',t}=0$. 
As the RUAVs initially store the global information $\mathcal{D}_{G,0}$, the incomplete center and complete center coincide at $t=0$, i.e., $\frac{1}{|\mathcal{I}_{0}|}\sum_{i'\in\mathcal{I}_{0}}\mathbf{p}_{i',0}=\frac{1}{|\widehat{\mathcal{I}}_{0}^i|}\sum_{i'\in\widehat{\mathcal{I}}^i_{0}}\widehat{\mathbf{p}}_{i',0}^i$. Moreover, the distance between $\widehat{\mathbf{p}}^i_{i',t}$ and $\mathbf{p}_{i',t}$ is bounded, since $\left\|\widehat{\mathbf{p}}^i_{i',t}-\mathbf{p}_{i',t}\right\|_2\le vt<vT$ always holds. Therefore, 
we can assume the following three mild conditions:
\begin{itemize}
	\item \textbf{Position bound:} $\left\|\widehat{\mathbf{p}}^i_{i',t}-\mathbf{p}_{i',t}\right\|_2\le b_1< vT$, $b_1>0$ is a constant;
	\item \textbf{Approximation of RUAV numbers:} $\frac{1}{|\mathcal{I}_t|}\approx\frac{1}{|\widehat{\mathcal{I}}_t^i|}$;
	\item \textbf{False RUAVs' bound:} $\left\|\frac{1}{|\widehat{\mathcal{I}}_t^i|}\sum_{i'\in\mathcal{I}^i_{r,t}\backslash\mathcal{I}_t}\widehat{\mathbf{p}}^i_{i',t}\right\|_2\le b_2$, $b_2>0$ is a constant.
\end{itemize}
Then the upper bound of the distance $\varpi_{i,t}$ between incomplete center and complete center can be calculated as \eqref{equ:center}, where $\mathcal{C}_{i,t}$ denotes the index set of RUAVs that are in the same RUAV cluster with RUAV$_{i,t}$. Hence, RUAVs using CR-MGCM nearly fly towards the same position as RUAVs using CR-MGCM$_{glob}$ at each time step. Besides, the inertia $\kappa$ in CR-MGCM can offer RUAVs the latest information of USNET to plan their trajectories. Therefore, CR-MGCM can reach the performance of CR-MGCM$_{glob}$ under the general UEDs.


\begin{strip}
	\hrule
	\begin{align}
		\label{equ:center_1}
		\varpi_{i,t}&=\left\|\sum_{i'\in\mathcal{I}_t\cap\widehat{\mathcal{I}}^i_{t}}\bigg(\frac{1}{|\mathcal{I}_t|}\mathbf{p}_{i',t}-\frac{1}{|\widehat{\mathcal{I}}_t^i|}\widehat{\mathbf{p}}_{i',t}^i\bigg)+\frac{1}{|\mathcal{I}_t|}\sum_{i'\in\mathcal{I}_t\backslash\mathcal{I}^i_{r,t}}\mathbf{p}_{i',t}-\frac{1}{|\widehat{\mathcal{I}}_t^i|}\sum_{i'\in\widehat{\mathcal{I}}^i_{t}\backslash\mathcal{I}_t}\widehat{\mathbf{p}}^i_{i',t}\right\|_2.
	\end{align}
	\hrule
	\begin{align}
		\label{equ:center}
		\varpi_{i,t}&=\left\|\frac{1}{|\mathcal{I}_t|}\sum_{i'\in\mathcal{I}_t}\mathbf{p}_{i',t}-\frac{1}{|\widehat{\mathcal{I}}_t^i|}\sum_{i'\in\widehat{\mathcal{I}}_t^i}\mathbf{p}_{i',t}^i\right\|_2=\left\|\sum_{i'\in\mathcal{I}_t\cap\widehat{\mathcal{I}}_t^i}\bigg(\frac{1}{|\mathcal{I}_t|}\mathbf{p}_{i',t}-\frac{1}{|\widehat{\mathcal{I}}_t^i|}\widehat{\mathbf{p}}_{i',t}^i\bigg)-\frac{1}{|\widehat{\mathcal{I}}_t^i|}\sum_{i'\in\widehat{\mathcal{I}}_t^i\backslash\mathcal{I}_t}\widehat{\mathbf{p}}^i_{i',t}\right\|_2\notag\\
		&\le\left\|\sum_{i'\in\mathcal{I}_t\cap\widehat{\mathcal{I}}_t^i\backslash\mathcal{C}_{i,t}}\frac{1}{|\mathcal{I}_t|}(\mathbf{p}_{i',t}-\widehat{\mathbf{p}}_{i',t}^i)\right\|_2+\left\|\frac{1}{|\widehat{\mathcal{I}}_t^i|}\sum_{i'\in\widehat{\mathcal{I}}_t^i\backslash\mathcal{I}_t}\widehat{\mathbf{p}}^i_{i',t}\right\|_2\le b_1+b_2.
	\end{align}
\hrule
\end{strip}
\section{Simulation Results}
\label{section:simulations}

In the simulation\footnote{{The source codes are available on \emph{https://github.com/nobodymx/resilient-swarm-communications-with-meta-graph-convolutional-networks}}}, the initial USNET consists of $N=200$ identical UAVs that are randomly distributed in a 1,000m$\times$1,000m$\times$100m three-dimensional space, as shown in \reffig{simulation:initial_state}. The parameters of UAVs are specified in Table \ref{UAV_param}, and the CLEC can be calculated as 
\begin{align}
&10\log_{10}\bigg(\frac{96\pi l_{ii',t}}{3}\bigg)+\frac{l_{ii',t}}{5}\exp\bigg(\frac{-l_{ii',t}^2-100}{10}\bigg)I_0(20l_{ii',t})\notag\\
&\approx 10\log_{10}\bigg(96l_{ii',t}\bigg)\le 40.62,
\end{align}
from which we can derive $l_{ii',t}\approx 120\text{m}$. Hence, the CLEC can be described as: any two distinct UAVs can establish a communication link if their distance is smaller than 120m. The period of the self-healing process is set to be 450 time steps, i.e., $T=450$. The number of GCLs in the GCN is $Q=8$.
\begin{figure*}[t]
	\begin{minipage}[t]{0.5\linewidth}
		\centering
		\includegraphics[width=90mm]{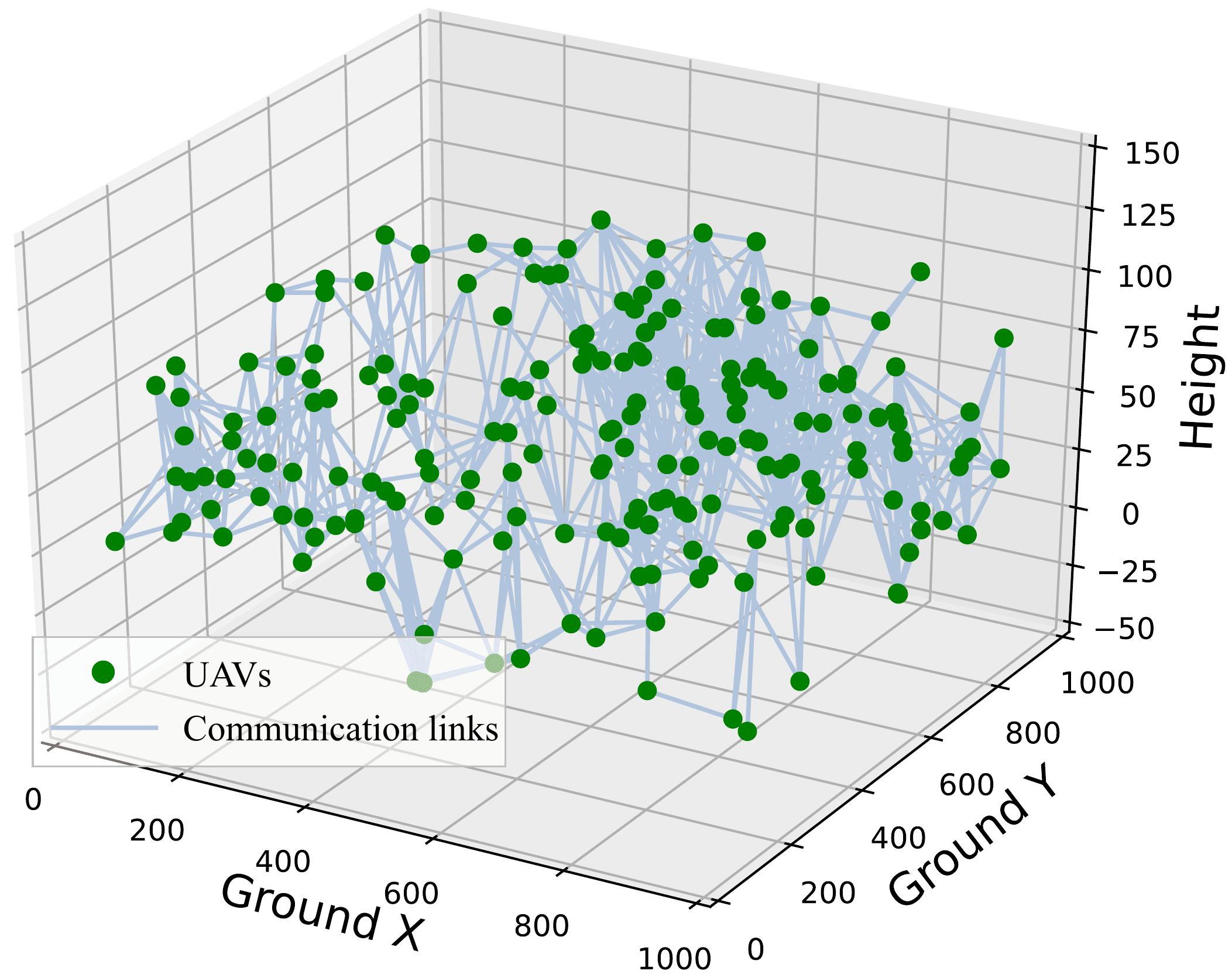} 
		\caption{{Initial distributions of the 200 identical UAVs .}}
		\label{simulation:initial_state}
	\end{minipage}
	\hspace{0ex}
	\begin{minipage}[t]{0.5\linewidth}
		\centering
		\includegraphics[width=93mm]{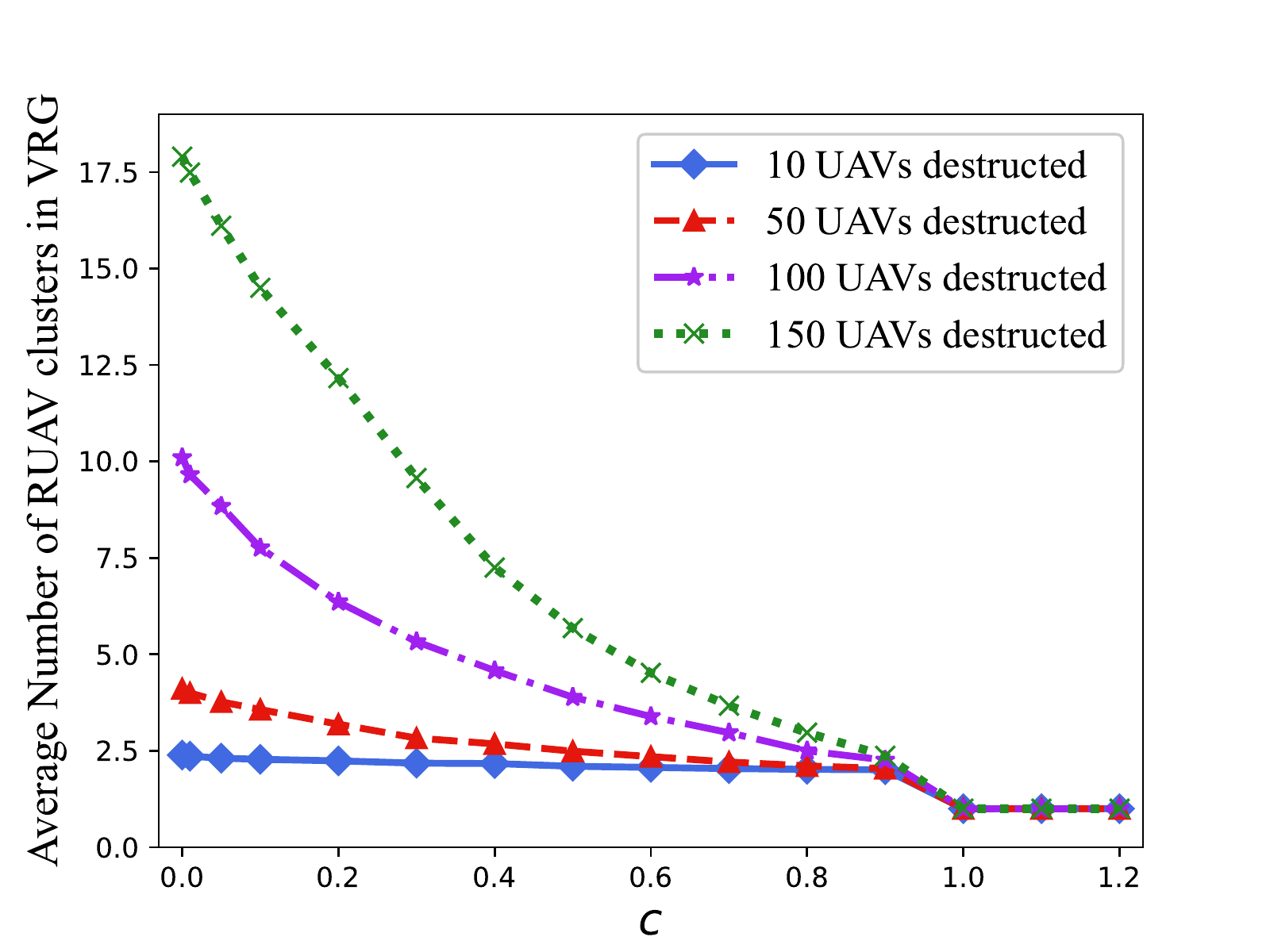} 
		\caption{The average number of RUAV clusters versus $c$.}
		\label{simulation:algorithm_1}
	\end{minipage}
\end{figure*}
\begin{figure*}[t]
	\begin{minipage}[t]{0.5\linewidth}
		\centering
		\includegraphics[width=89.7mm]{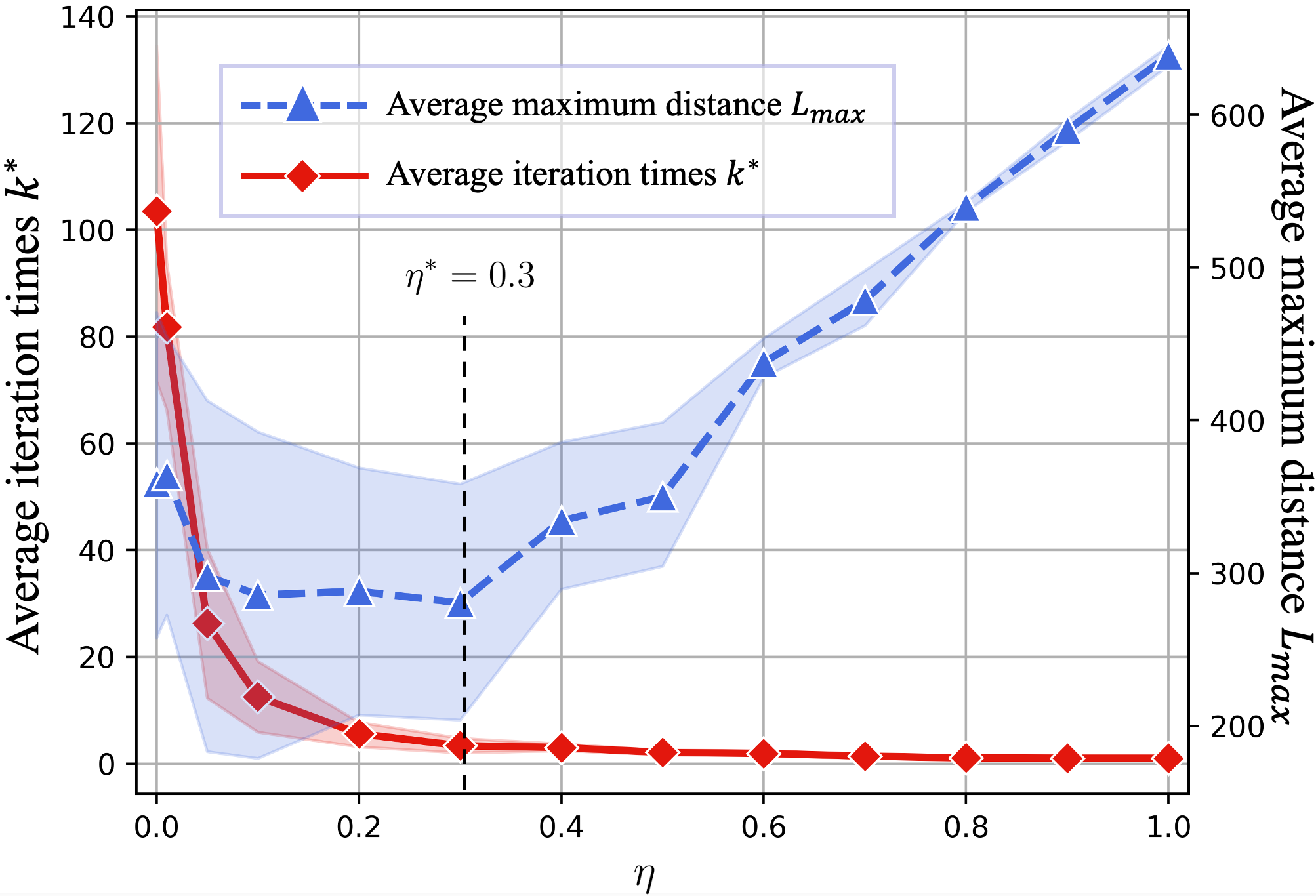} 
		\caption{The average of $k^*$ versus $\eta$ and the average of $L_{max}$ versus $\eta$. The $\epsilon$ is $1$.  }
		\label{simulation:eta*}
	\end{minipage}
	\hspace{0ex}
	\begin{minipage}[t]{0.5\linewidth}
		\centering
		\includegraphics[width=92.2mm]{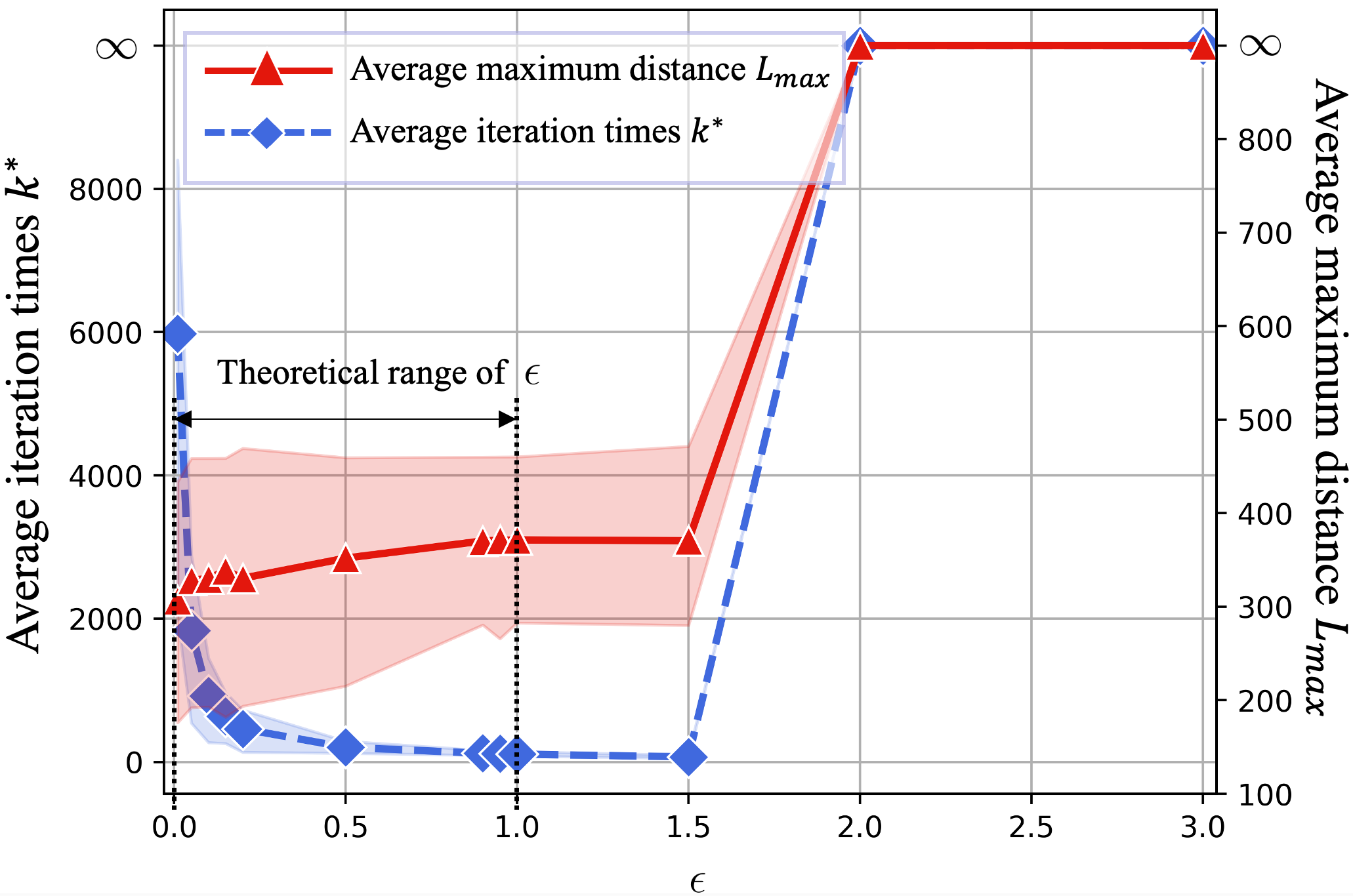} 
		\caption{The average of $k^*$ versus $\epsilon$, and the average of $L_{max}$ versus $\epsilon$. The $\eta$ is $0.3$.  }
		\label{simulation:K*}
	\end{minipage}
\end{figure*} 

\subsection{Verifications of Algorithm 1}

Express the virtual distance $d_t^v$ in the VRG as $d_t^v=120+c(d_{min,t}^v-120)$, where $d_{min,t}^v$ is obtained by Algorithm \ref{algorithm:find_d_min} and $c\ge 0$ is a coefficient. We randomly destroy 10, 50, 100 and 150 UAVs of the initial USNET 100 times each, and the average number of RUAV clusters in the VRG versus $c$ is shown in \reffig{simulation:algorithm_1}. When $c=0$ and $d_t^v=120$m, the average number of RUAV clusters is bigger than 1 and the VRG cannot form CCNs. As $c$ gets closer to 1, the virtual distance $d_t^v$ becomes larger and the average number of RUAV clusters in the VRG decreases. The VRG cannot form a CCN until $c=1$ and $d_t^v=d^v_{min,t}$. Hence, Algorithm \ref{algorithm:find_d_min} can guarantee to find the minimal virtual distance $d_{min,t}^v$ that makes the VRG a CCN.

\subsection{Finding $\eta^\star$ and $\epsilon^\star$ of the CR-MGC}
\label{subsection_simulation:find_best_hyperparameter}
We randomly {destruct} 10, 50, 100, 150 UAVs of the initial USNET 100 times each. \reffig{simulation:eta*} shows the average of the number of GCO $G(\cdot)$ iterations $k^*$ versus $\eta$. The average of $L_{max}$ versus $\eta$ is also shown in Fig.~7.
We can see that the average of $k^*$ drops with the increase of $\eta$, while the average of $L_{max}$ slightly decreases when $\eta\le0.3$ and continuously increases when $\eta>0.3$. Hence, we choose $\eta^\star=0.3$ as the best value of $\eta$ to balance both $k^*$ and $L_{max}$.

We randomly {destruct} 10, 50, 100, 150 UAVs of the initial USNET 100 times each. \reffig{simulation:K*}  shows the average of $k^*$ versus $\epsilon$. The average of $L_{max}$ versus $\epsilon$ is also shown in \reffig{simulation:K*} .
When $\epsilon\in[0,1.5]$, the average of $k^*$ drops with the increase of $\epsilon$, while the average of $L_{max}$ increases. However, when $\epsilon>1.5$, the GCO diverges and both the average of $k^*$ and average of $L_{max}$ go to infinity. Recall that $H_t=\frac{\epsilon}{\left\|\mathbf{A}^v_t\right\|_\infty}$ and the theoretical range of $K_t$ is $0< H_t\le\frac{1}{\left\|\mathbf{A}^v_t\right\|_\infty}$ (or equivalently $0<\epsilon\le 1$). Hence, the results in \reffig{simulation:K*} verify the correctness of the theoretical range of $H_t$. We can choose $\epsilon^\star=1$ as the best value of $\epsilon$ to balance both $k^*$ and $L_{max}$.
\subsection{Meta Learning of the GCN}
\begin{figure*}[t]
	\begin{minipage}[t]{0.5\linewidth}
		\centering
		\includegraphics[width=89mm]{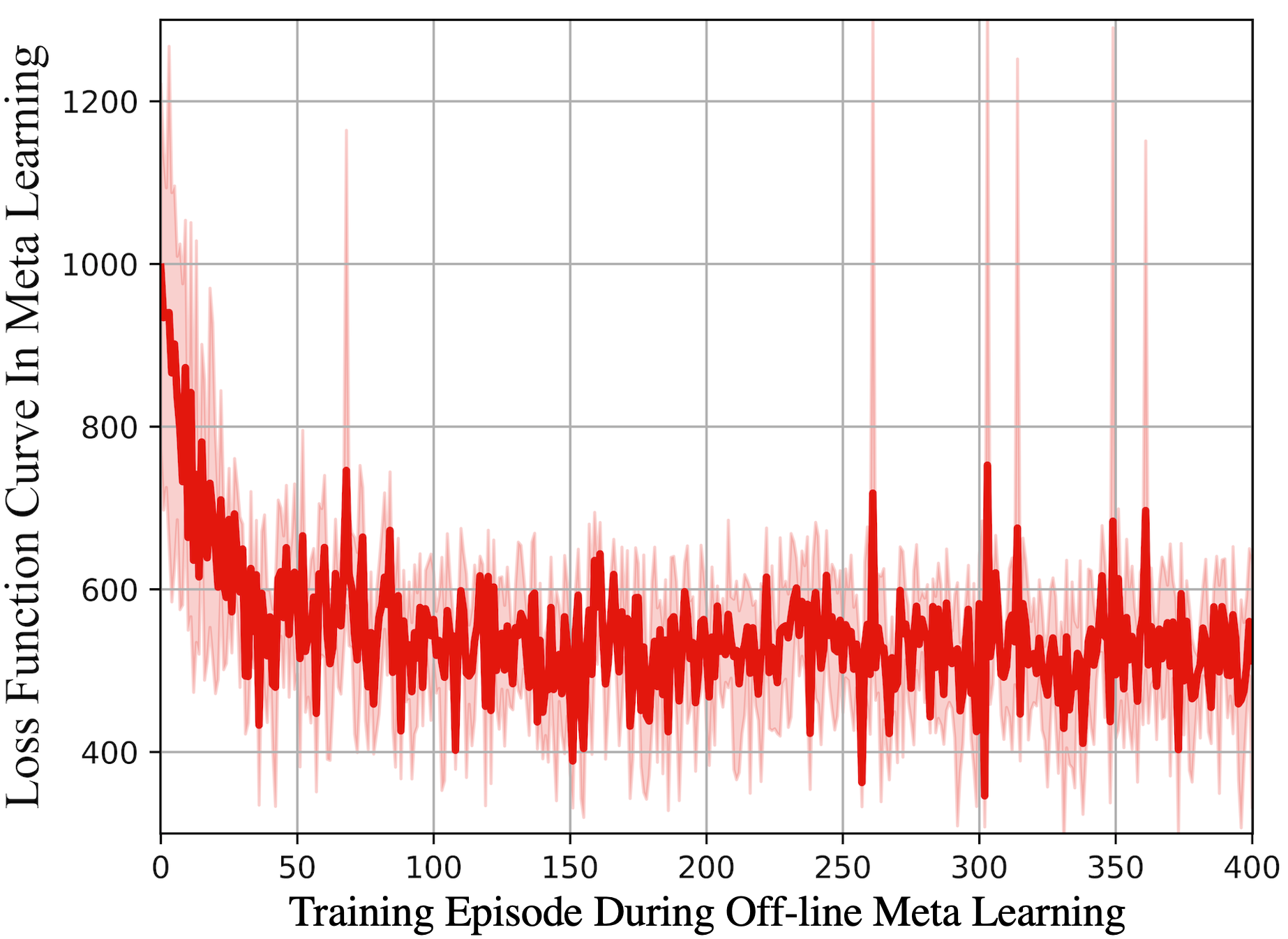} 
		\caption{Loss function curve of the mGCNs during meta learning. }
		\label{simulation:meta_train}
	\end{minipage}
	\hspace{0ex}
	\begin{minipage}[t]{0.5\linewidth}
		\centering
		\includegraphics[width=89.0mm]{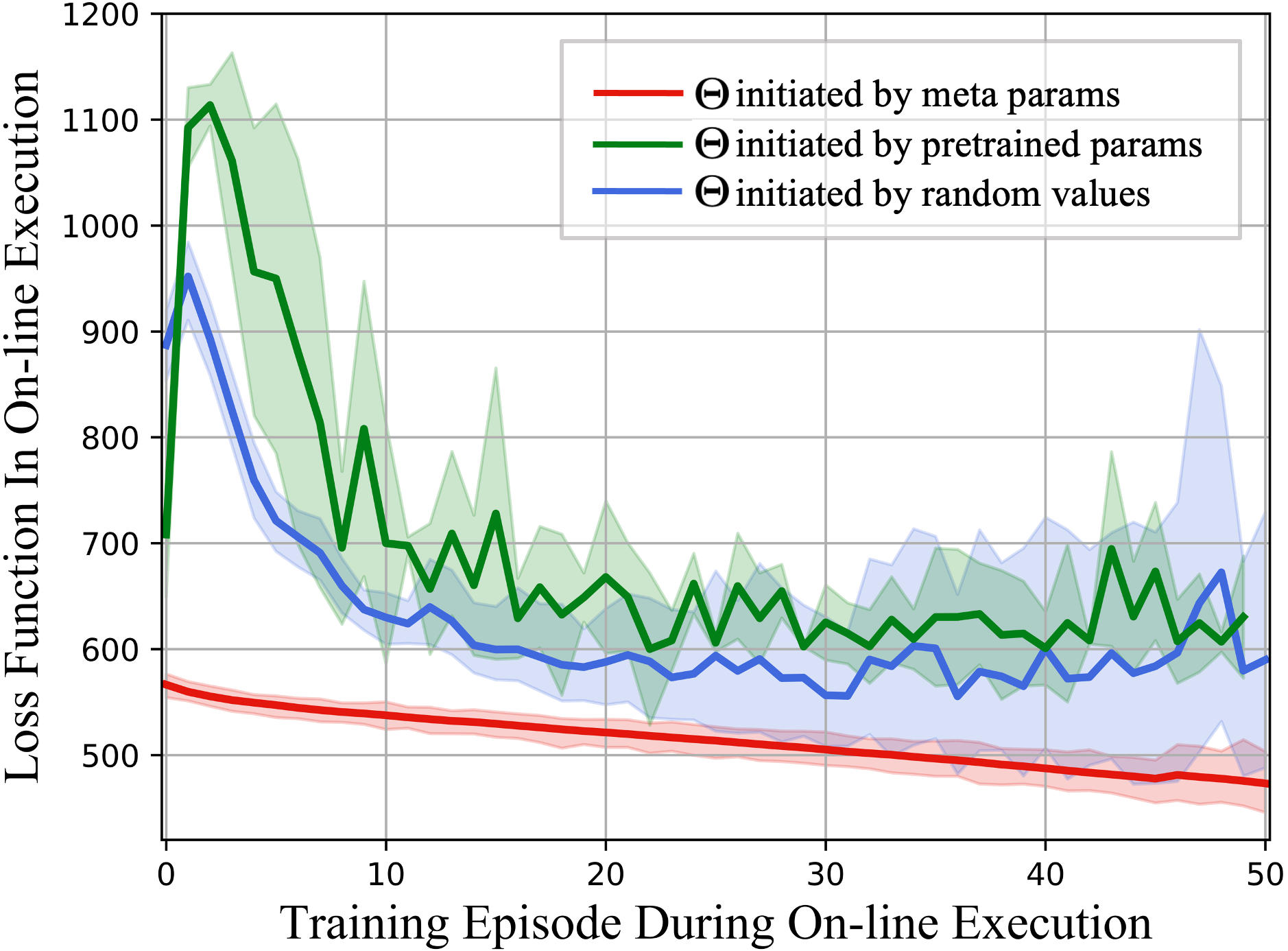} 
		\caption{Loss function during the on-line executions, $M=50$.  }
		\label{simulation:meta_test}
	\end{minipage}
\end{figure*}
We build $199$ mGCNs since the initial USNET contains $200$ UAVs. For each mGCN, we construct a support set and query set with $U_0=400$ topology matrices each. 
\reffig{simulation:meta_train} shows the average loss function curve of all mGCNs during the off-line meta learning. We can see that the loss function starts from 1000 and drops stably to 500 during the off-line meta learning. The consistent decrease of the loss function indicates that the parameters of the mGCNs are gradually moving to better values.

\reffig{simulation:meta_test} shows the loss function curves of the GCN during the training process in on-line executions, where the parameters of GCN $\mathbf{\Theta}$ are initiated by the meta parameters $\mathbf{\Gamma}^\star_{k}$, the pre-trained parameters, and random values, respectively. We set the on-line training episode to be $M=50$. On the one hand, the loss function curve of GCN initiated by meta parameters starts from 570 that is smaller than other two curves (700 and 900, respectively). This means that the meta parameters are better initialized values than both the pre-trained parameters and random parameters. On the other hand, the loss function curve of GCN initiated by meta parameters decreases continuously during the on-line training process and reaches lower values than other two curves, which implies the meta parameters have great potential in performance.  
\begin{figure}
	\centering
	\includegraphics[width=90mm]{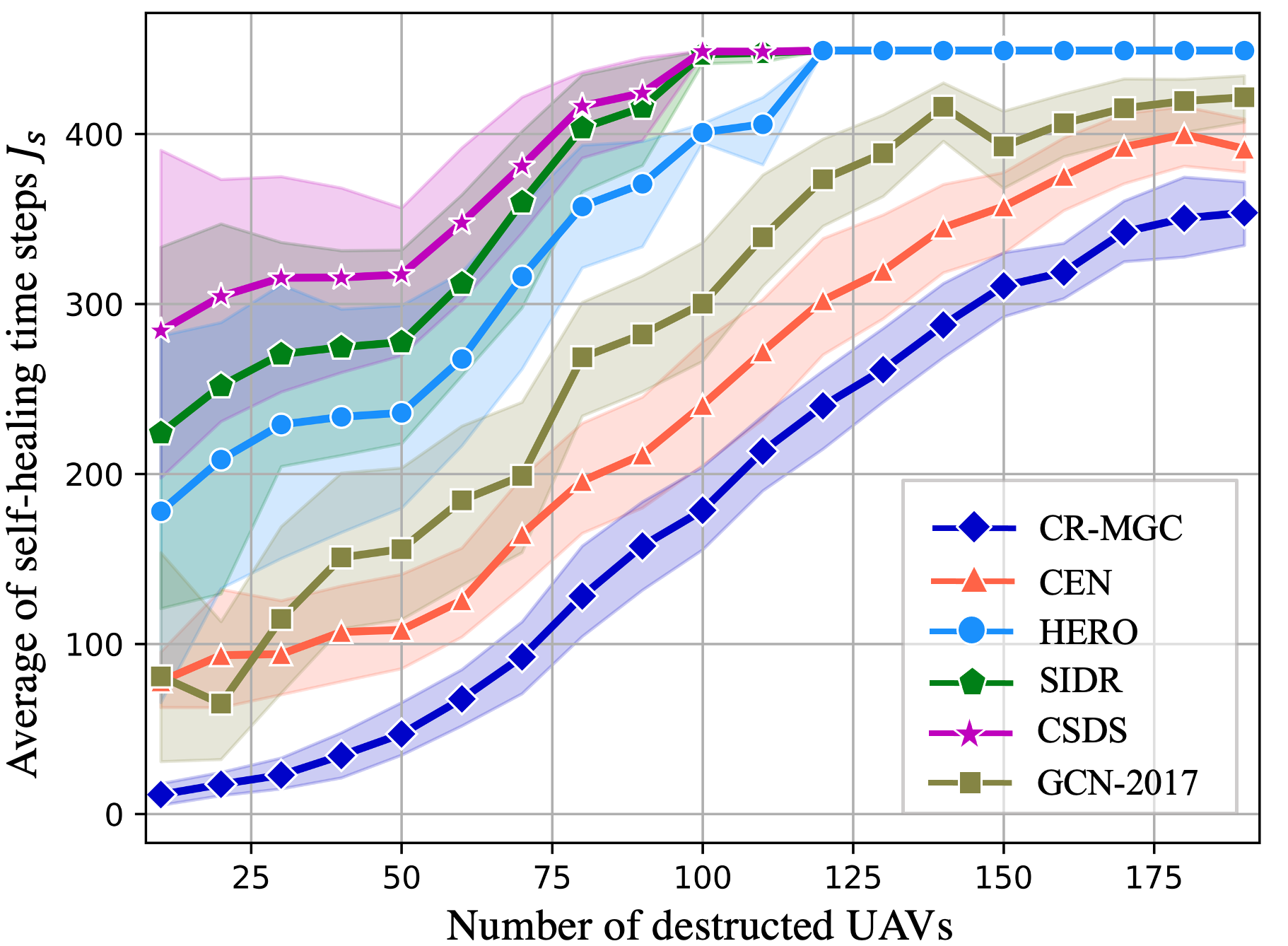}
	\caption{Average self-healing time steps ${J}_s$ under different number of destructed UAVs. }
	\label{fig:single_time_UED}
\end{figure}
\subsection{SCC of One-off UEDs in $(\mathbf{P1})$}
\reffig{fig:single_time_UED} shows the average self-healing time steps ${J}_s$ of the CR-MGC under one-off UEDs. The performances of HERO\cite{hero}, SIDR\cite{sidr}, CSDS\cite{csds}, GCN-2017\cite{gcn_2017}, and CEN\footnote{CEN represents the algorithm that makes each RUAV fly to the center of their positions directly.} are also displayed for comparisons. We randomly {destruct} 10, 20, 30, ..., 190 UAVs of the initial USNET 100 times each, and take the average value of the self-healing time to plot the curves of different algorithms. The shaded areas represent the 100\% confidential intervals of the average self-healing time. We can see that with the increase of the number of {destructed} UAVs, the self-healing time of all the algorithms increases. Moreover, the average self-healing time of the CR-MGC is smaller than those of other four algorithms under any number of {destructed} UAVs. Hence, the CR-MGC can rebuild the communication connectivity of the USNET within shorter time. 

\begin{figure*}[t]
	\setlength{\abovecaptionskip}{0.5cm}
	\centering
	\subfigure[The initial USNET is {destructed} into 11 RUAV clusters.]{
		\label{fig_simulations:3_3_1} 
		\includegraphics[width=88mm]{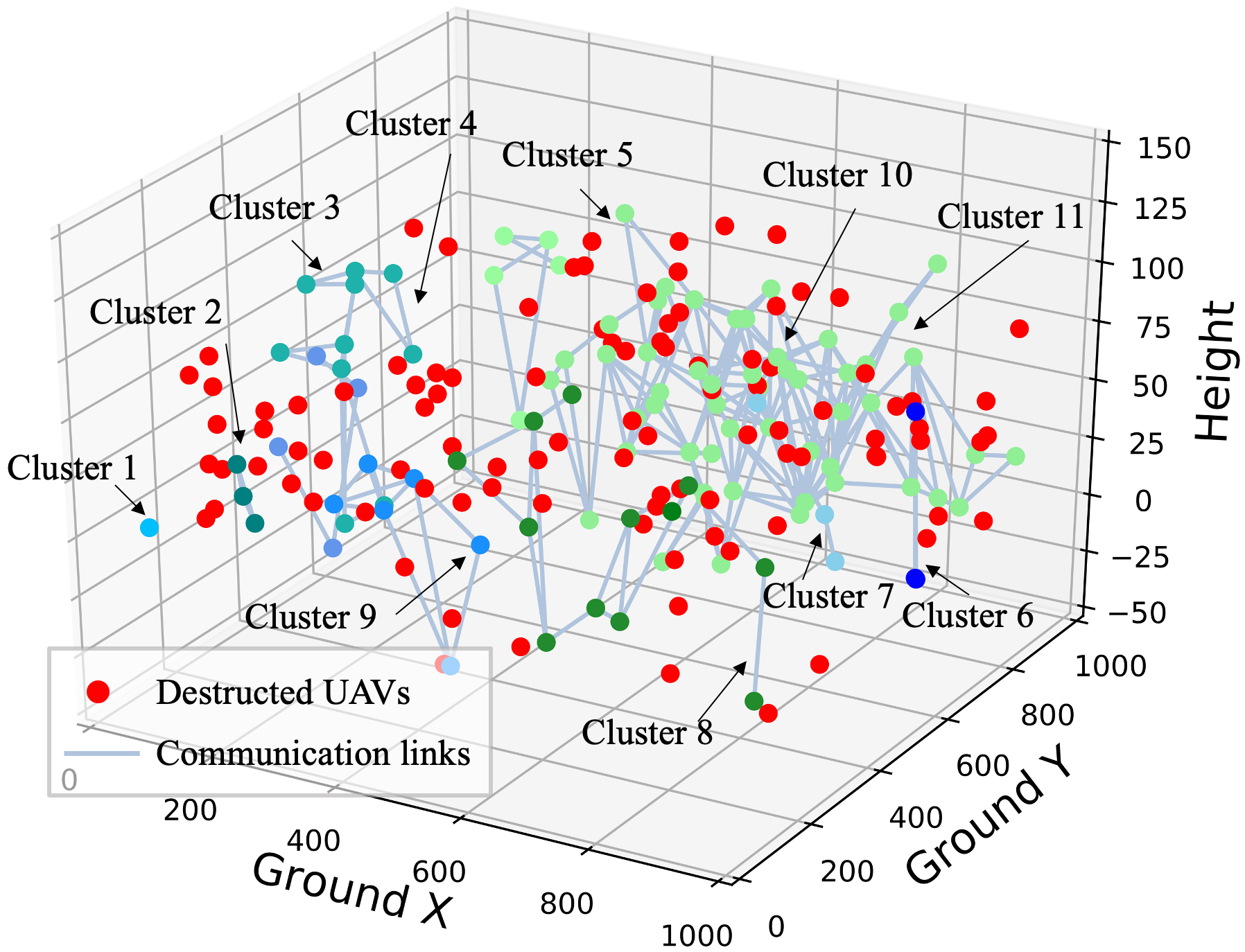}}
	\subfigure[The changing of the positions when applying the GCO $G(\cdot)$.]{
		\label{fig_simulations:3_3_2} 
		\includegraphics[width=88mm]{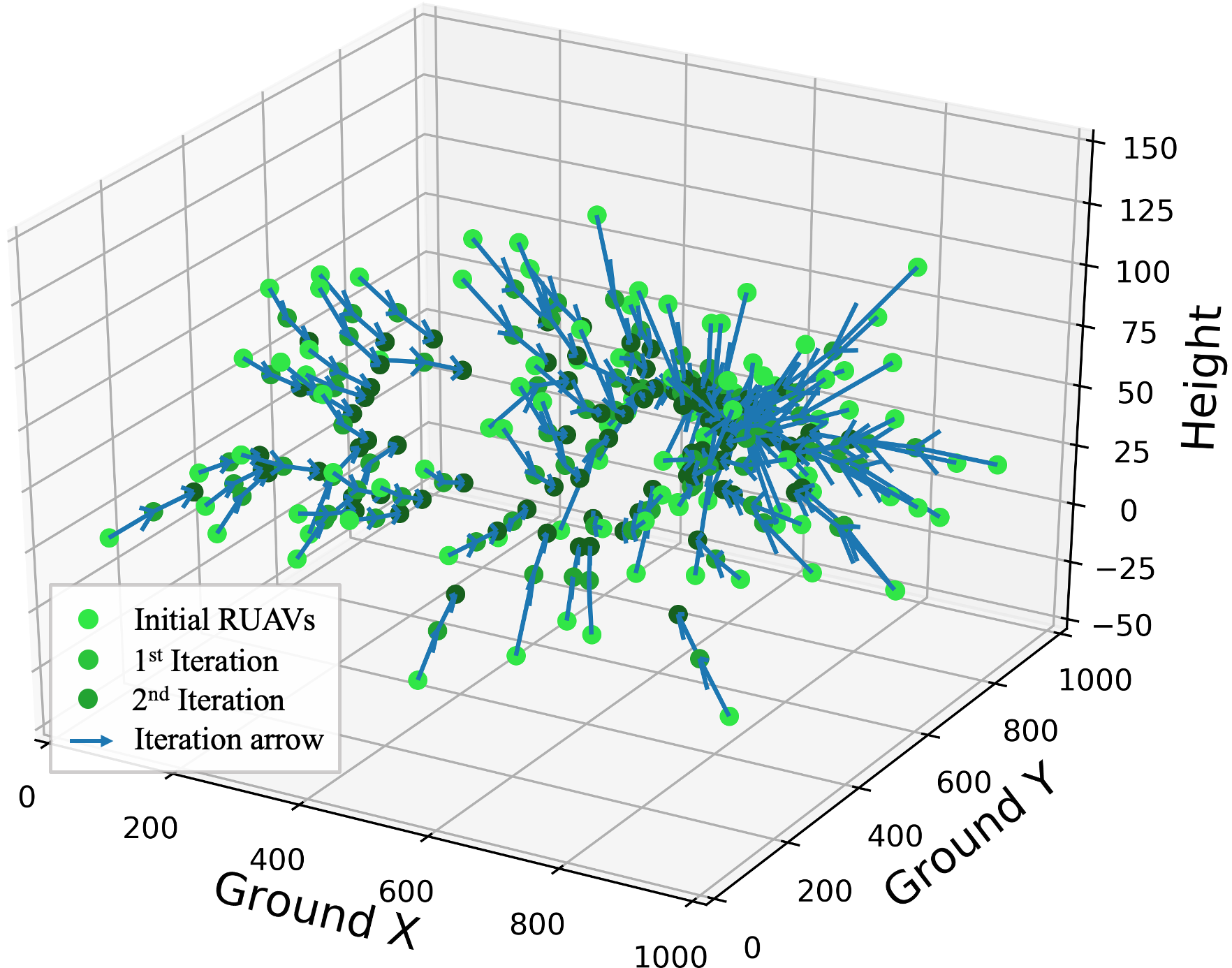}}

	\subfigure[Flying trajectories of RUAVs using CR-MGC.]{
		\label{fig_simulations:3_3_3} 
		\includegraphics[width=88mm]{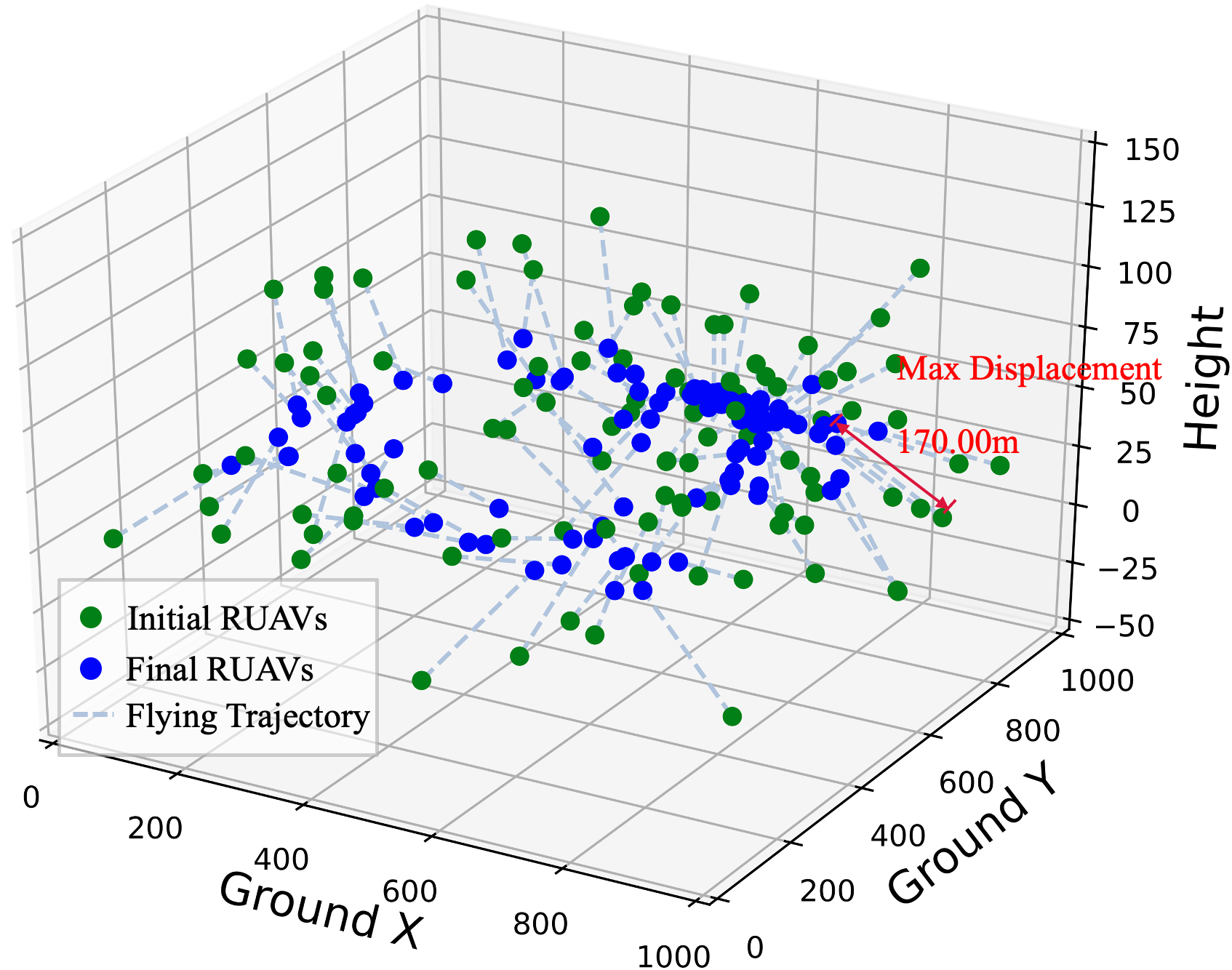}}
	\subfigure[The number of RUAV clusters versus the time steps.]{
		\label{fig_simulations:3_3_4} 
		\includegraphics[width=88mm]{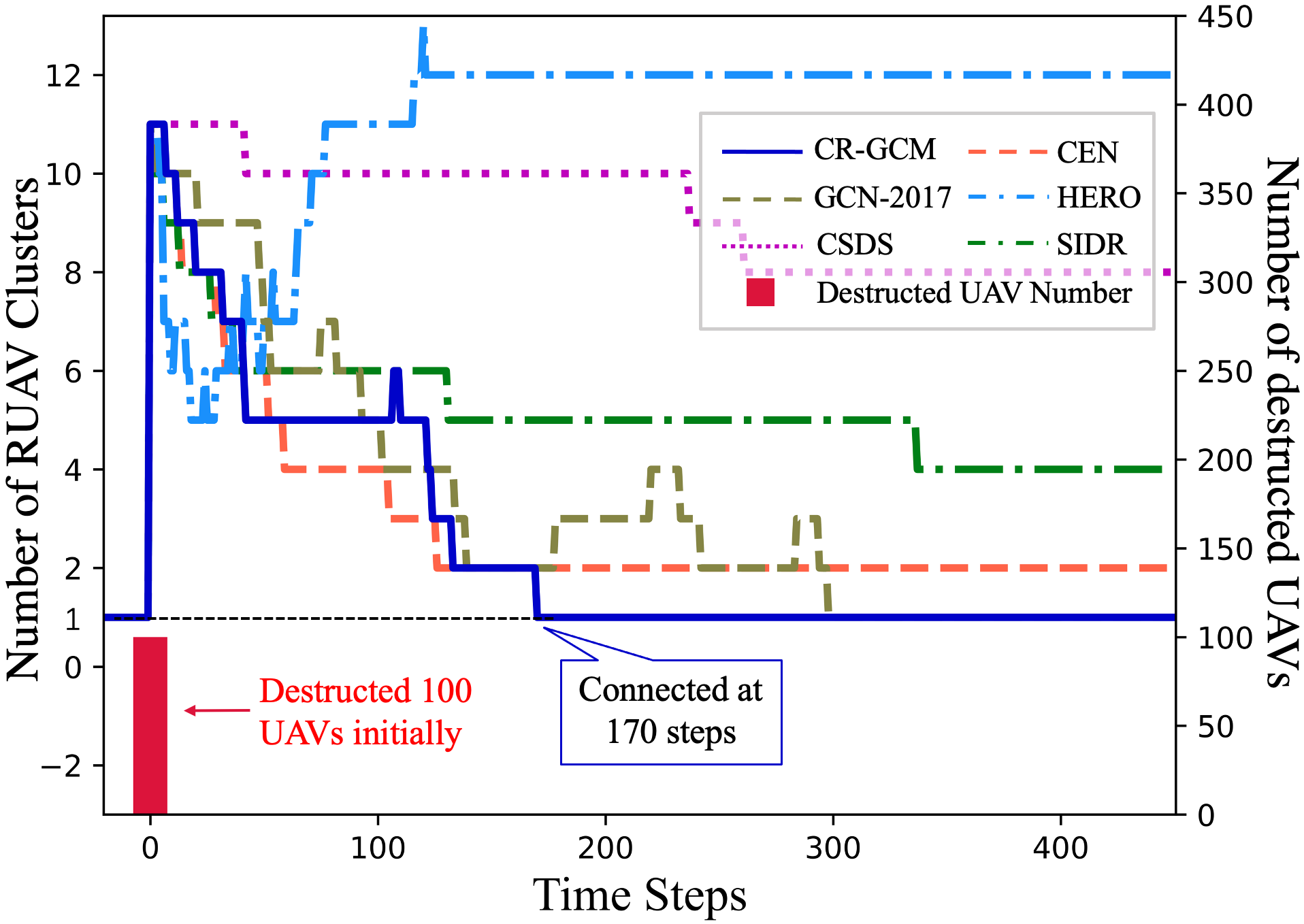}}
	
	\caption{Disruptions to the initial USNET and the self-healing process under different algorithms.}
	\label{fig:3-3}
\end{figure*}

\reffig{fig:3-3} shows the trajectories of the RUAVs during a certain self-healing process\footnote{{Note that the motion graphs of the self-healing process are available on \emph{https://github.com/nobodymx/resilient-swarm-communications-with-meta-graph-convolutional-networks}}}, where the one-off UED destroys $100$ UAVs at $t=0$. As shown in \reffig{fig_simulations:3_3_1}, the initial USNET is {destructed} into $C_0=11$ RUAV clusters, where nodes with the same color denotes the RUAVs in the same RUAV cluster. \reffig{fig_simulations:3_3_2} shows that the GCOs can make the RUAVs gather towards their center to form a CCN, which is consistent with Proposition \ref{propositions_2}. \reffig{fig_simulations:3_3_3} shows the flying trajectory of each RUAV using CR-MGC.
The maximum displacement of all RUAVs is 170m.  
\reffig{fig_simulations:3_3_4} shows that the number of RUAV clusters of the RUAV graph $\mathcal{G}_t$ decreases with CR-MGC. Moreover, the CR-MGC makes the RUAVs form a CCN within the least time steps.  

\subsection{SCC of General UEDs in $(\mathbf{P2})$}
\begin{figure*}[t]
	\begin{minipage}[t]{0.5\linewidth}
		\centering
		\includegraphics[width=90.2mm]{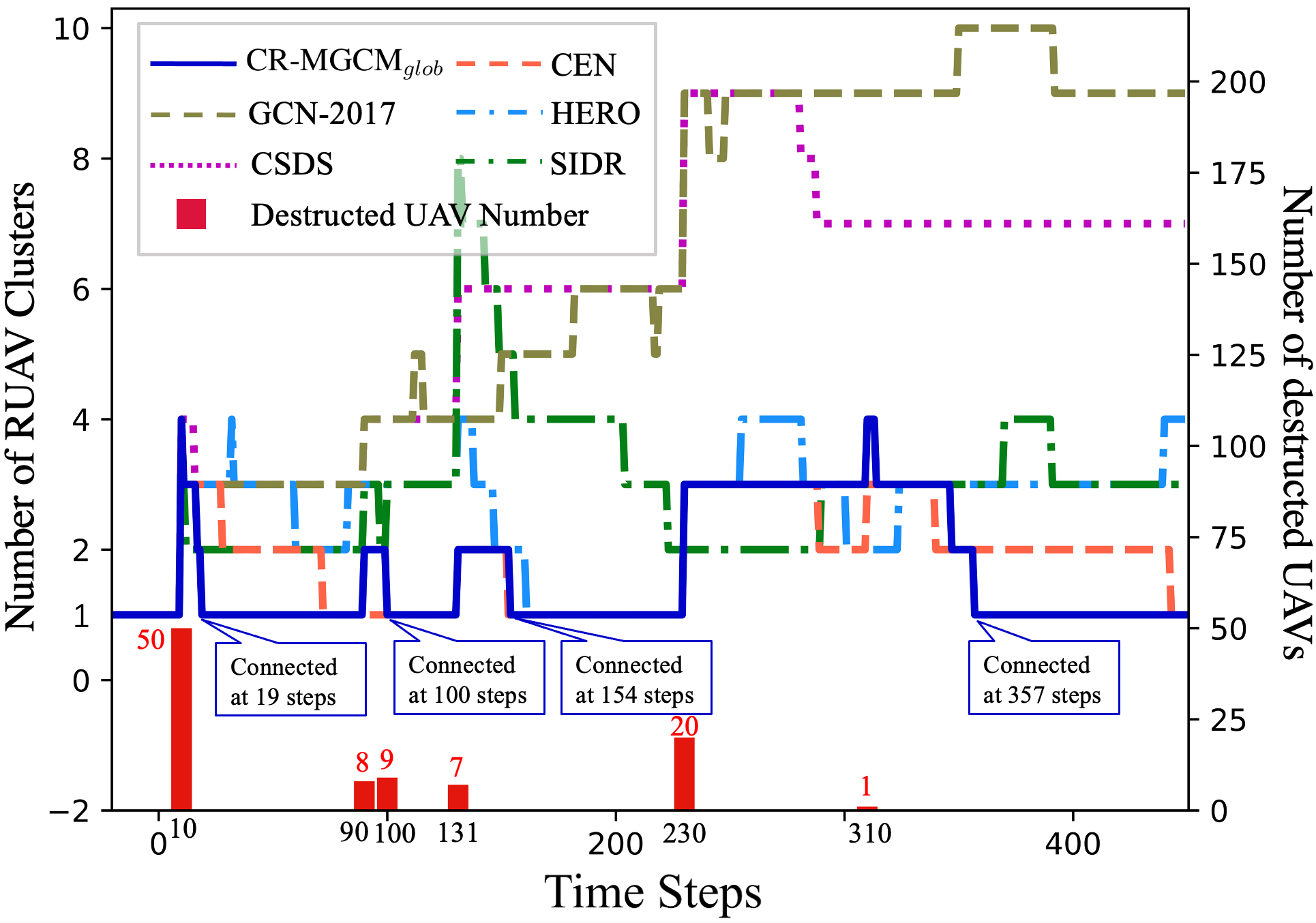} 
		\caption{The number of RUAV cluster versus time steps under the general UEDs with global information.}
		\label{simulation:4_1}
	\end{minipage}
	\hspace{0.5ex}
	\begin{minipage}[t]{0.5\linewidth}
		\centering
		\includegraphics[width=90.2mm]{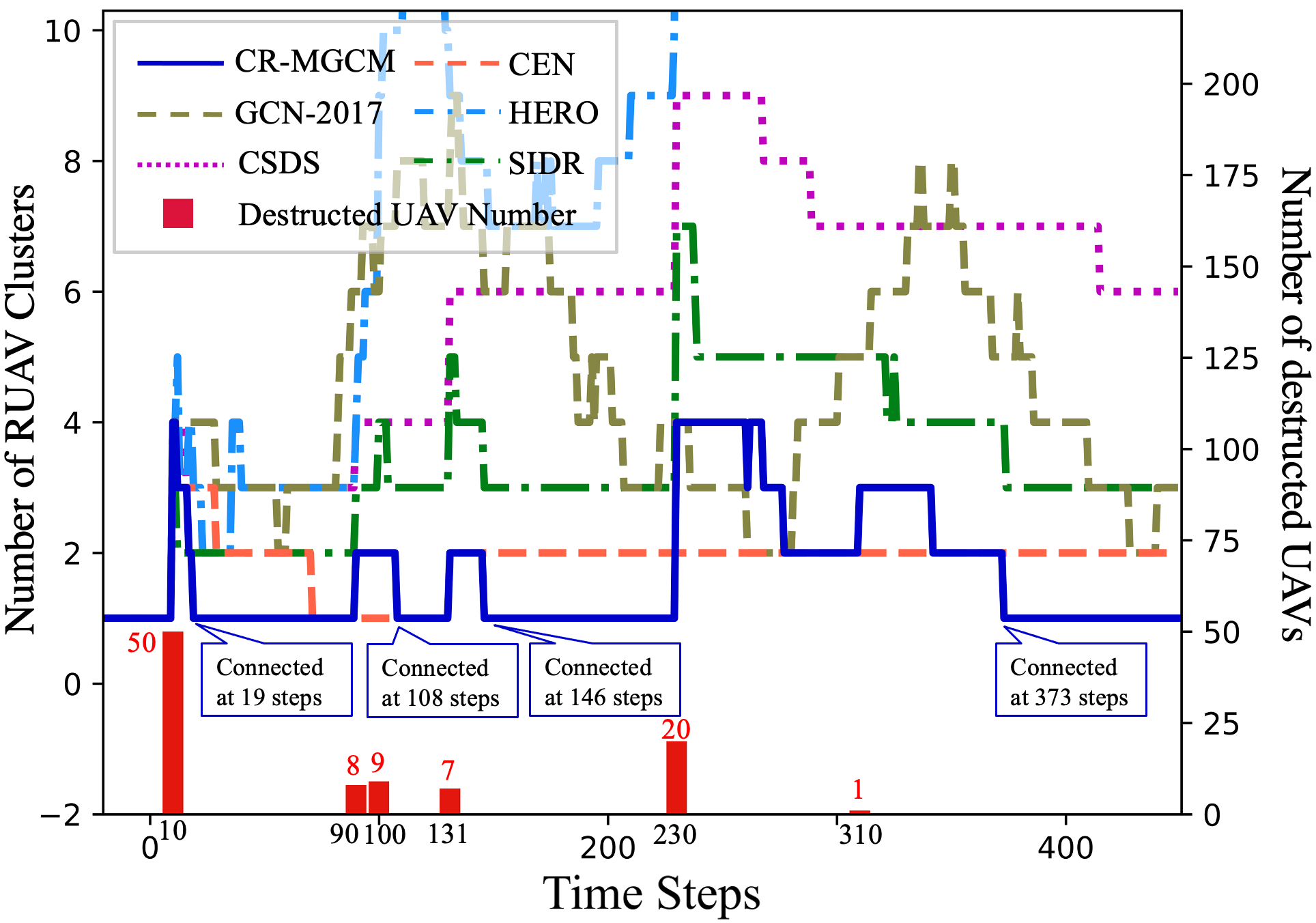} 
		\caption{The number of RUAV cluster versus time steps under the general UEDs with monitoring mechanisms.}
		\label{simulation:4_2}
	\end{minipage}
\end{figure*}

\begin{figure*}[t]
	\begin{minipage}[t]{0.5\linewidth}
		\centering
		\includegraphics[width=89.2mm]{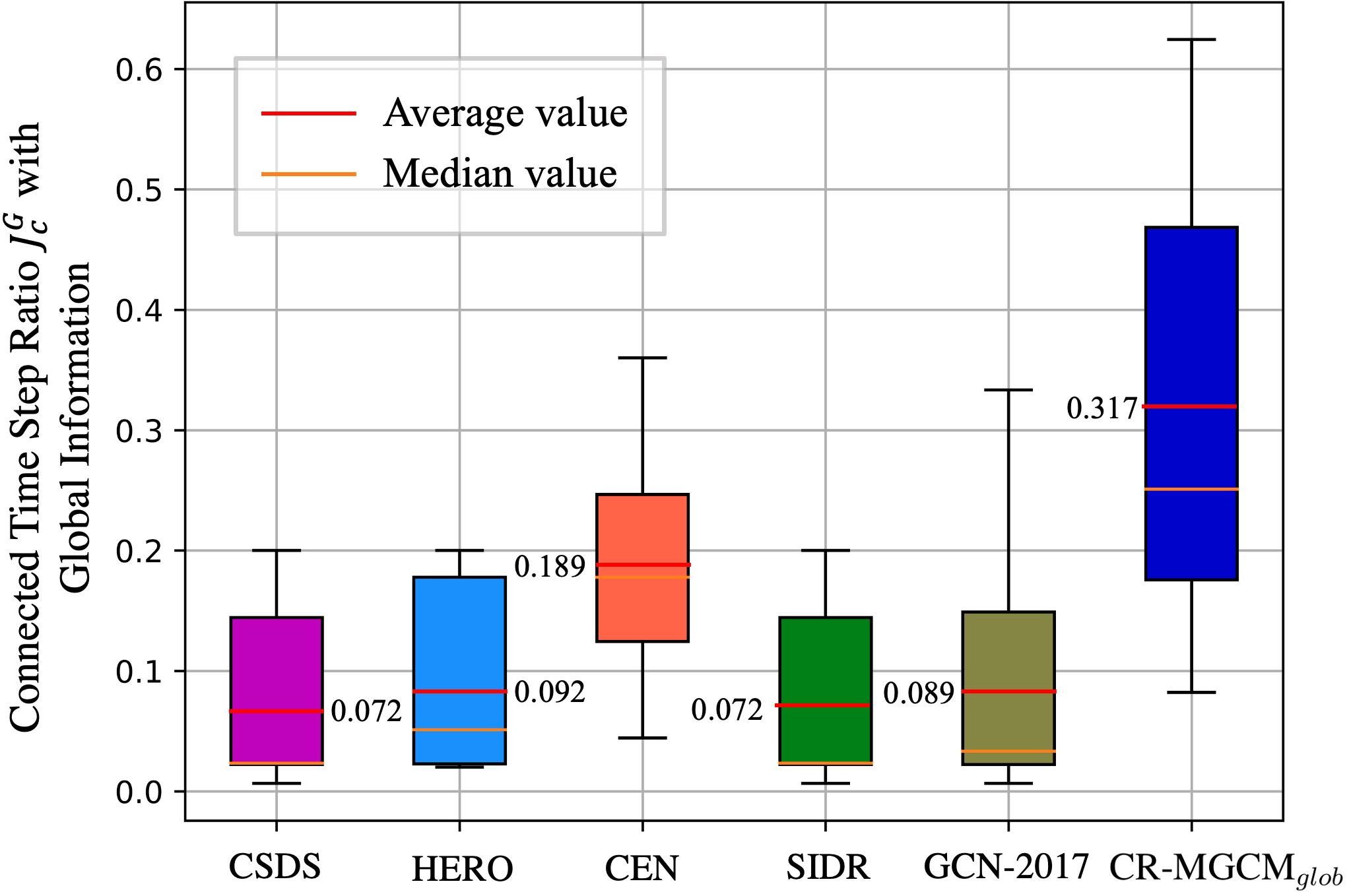} 
		\caption{The connected time step ratio ${J}_c^G$ of different algorithms with global information.  }
		\label{simulation:4_3}
	\end{minipage}
	\hspace{0.5ex}
	\begin{minipage}[t]{0.5\linewidth}
		\centering
		\includegraphics[width=89.2mm]{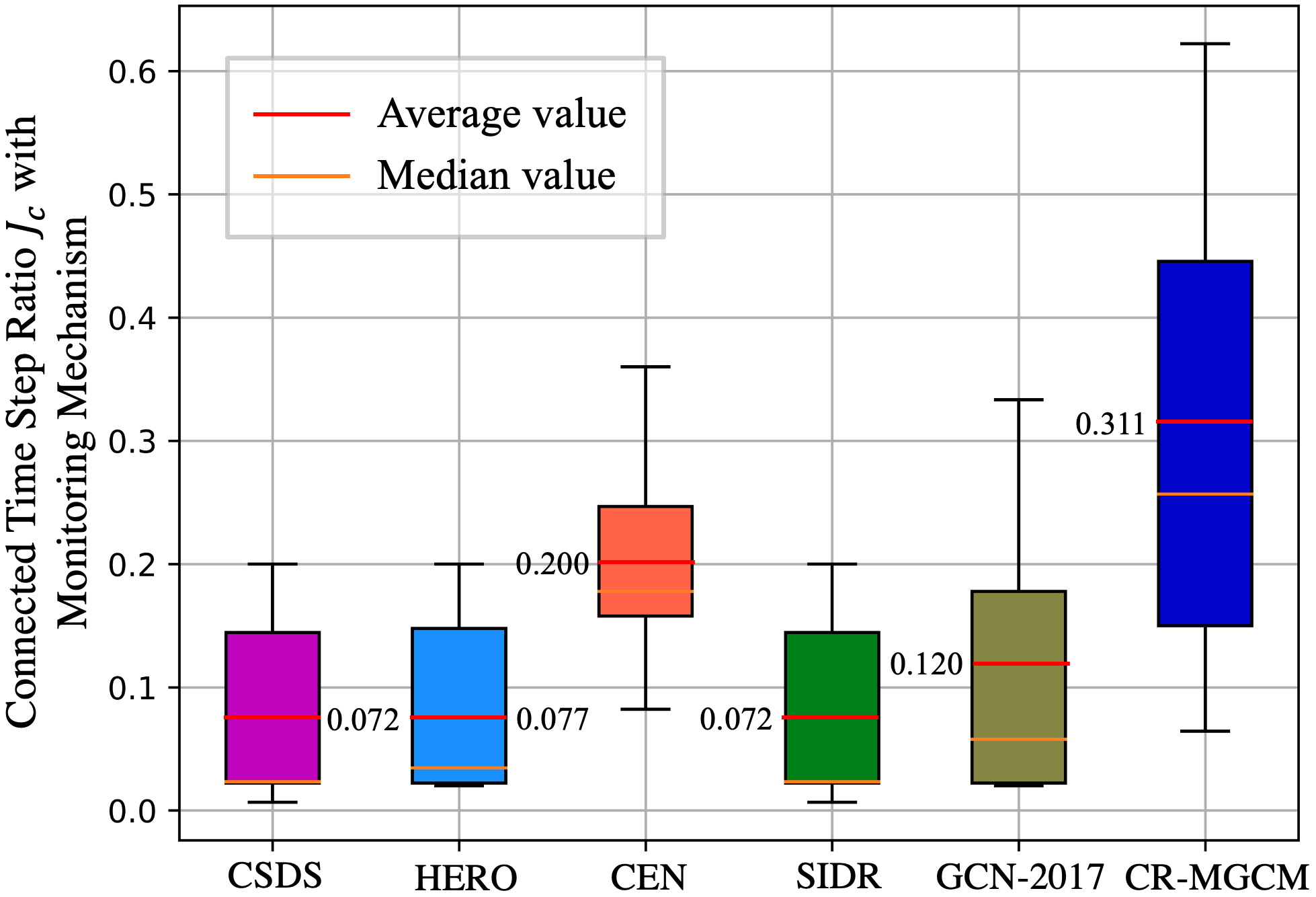} 
		\caption{The connected time step ratio ${J}_c$ of different algorithms with monitoring mechanisms.  }
		\label{simulation:4_4}
	\end{minipage}
\vspace{-0.2cm}
\end{figure*}

\reffig{simulation:4_1} and \reffig{simulation:4_2} both show the number of RUAV clusters $C_t$ using different algorithms under the same general UED. 
However, the RUAVs in the simulation of \reffig{simulation:4_1} have global information at each time step, while the RUAVs in the simulation of \reffig{simulation:4_2} do not and can only utilize the  monitoring mechanism.
The UED happens at 10, 90, 100, 131, and 230 time step and destruct 50, 8, 9, 7, and 20 UAVs, respectively. We can see that the RUAVs using CR-MGCM$_{glob}$ and CR-MGCM both quickly forms a CCN after each UED, while the RUAVs using other algorithms slowly forms a CCN after UEDs or even cannot form CCNs.  Hence, CR-MGCM$_{glob}$ and CR-MGCM can effectively rebuild the communication connectivity of the USNET within shorter time steps than the existing algorithms.


We {destructed} the USNET with 10 distinct general UEDs and depict the distribution of the connected time step ratio by boxplots shown in \reffig{simulation:4_3} and \reffig{simulation:4_4}. The RUAVs in the simulation of \reffig{simulation:4_3} have global information at each time step, while the RUAVs in the simulation of Fig.~16 only utilize the monitoring mechanism. In order to distinguish from $J_c$, we denote the connected time step ratio in \reffig{simulation:4_3} as $J_c^G$.  We can see that the average ${J}_c^G$ with CR-MGCM$_{glob}$ is larger than that of other algorithms, which indicates the effectiveness of the CR-MGCM$_{glob}$ under the general UEDs. We can also see that the average ${J}_c$ with CR-MGCM is larger than that of other algorithms, which indicates the effectiveness of the CR-MGCM under the general UEDs. Moreover, the ratio between the average $J_c$ with CR-MGCM and the average $J_c^G$ with CR-MGCM$_{glob}$ is $\frac{{J}_c}{{J}_c^G}=\frac{0.311}{0.317}=98.11\%$, which indicates that CR-MGCM can reach the performance of CR-MGCM$_{glob}$ under the general UEDs.

\begin{figure}[t]
	\centering
	\includegraphics[width=88mm]{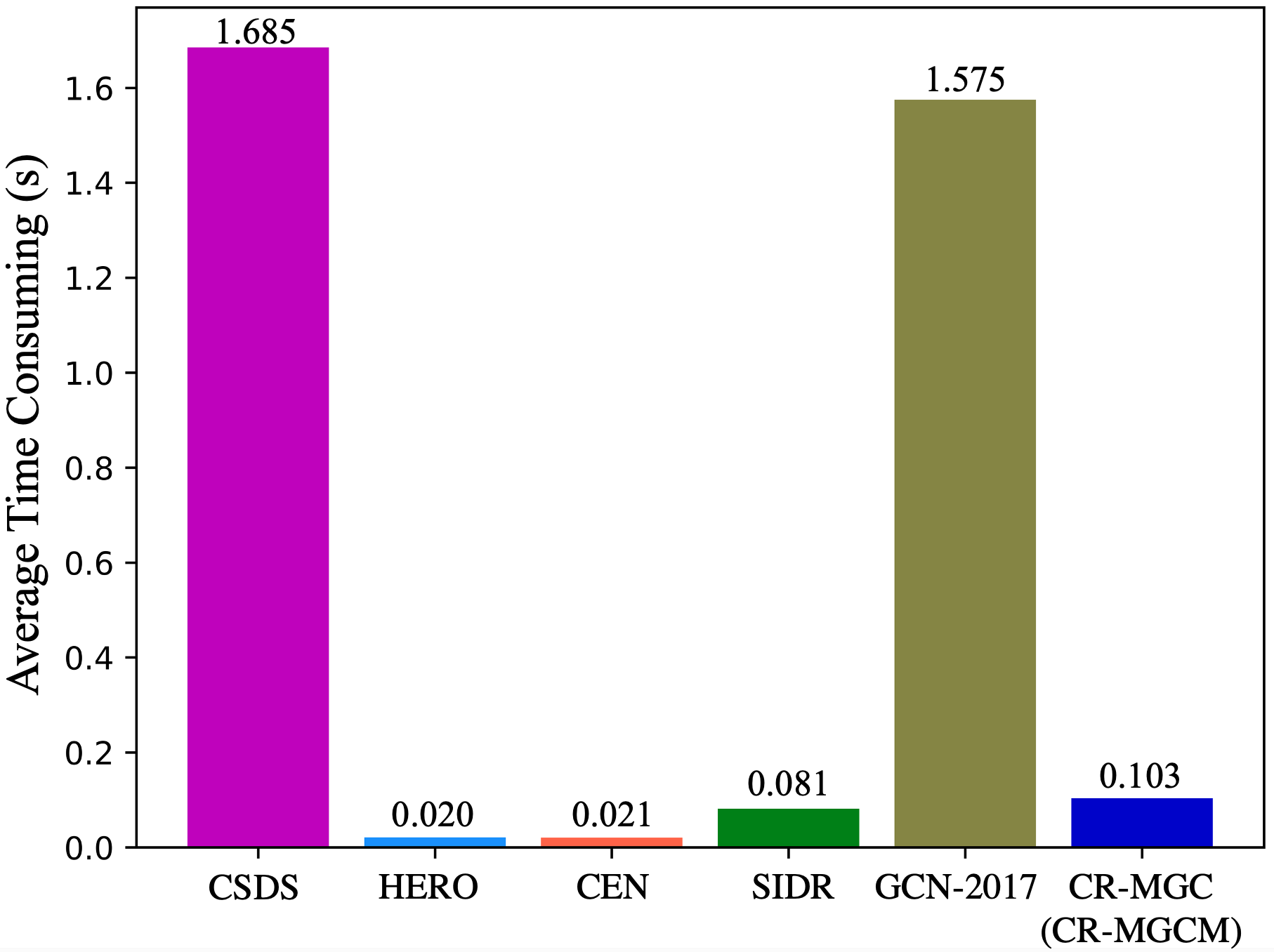}
	\caption{Average time consumptions of on-line executions {with} different algorithms. }
	\label{fig:time_consuming}
	\vspace{-0.5cm}
\end{figure}
\subsection{Time Consuming Comparisons}
\reffig{fig:time_consuming} compares the average on-line execution time cost at one time step of different algorithms. We can see that the average time cost of CR-MGC have the same magnitudes with HERO, CEN and SIDR, but is much smaller than CSDS and GCN-2017.  Note that the CR-MGCM and CR-MGCM$_{glob}$ both have the same time costs with CR-MGC, since they use the same GCN structures. This indicates that CR-MGC, CR-MGCM and CR-MGCM$_{glob}$ have acceptable on-line execution time costs.

\section{Conclusion}
\label{section:conclusions}
In this paper, we studied the SCC problem of the USNET under one-off UEDs and general UEDs. Specifically, we proposed a CR-MGC algorithm to cope with the SCC problem under one-off UEDs and verify its convergence. We also developed a meta learning scheme to improve the on-line executions of CR-MGC. For the SCC problem under the general UEDs, we designed the CR-MGCM algorithm to plan the trajectories of RUAVs. Numerical results showed that the proposed algorithms can rebuild the communication connectivity of the USNET within shorter time than the existing algorithms under both one-off UEDs and general UEDs. The experiment results also showed that the meta learning scheme could not only enhance the performance of the proposed algorithms, but also reduce the on-line execution time costs of them.
\begin{figure}[t]
	\centering
	\includegraphics[width=87mm]{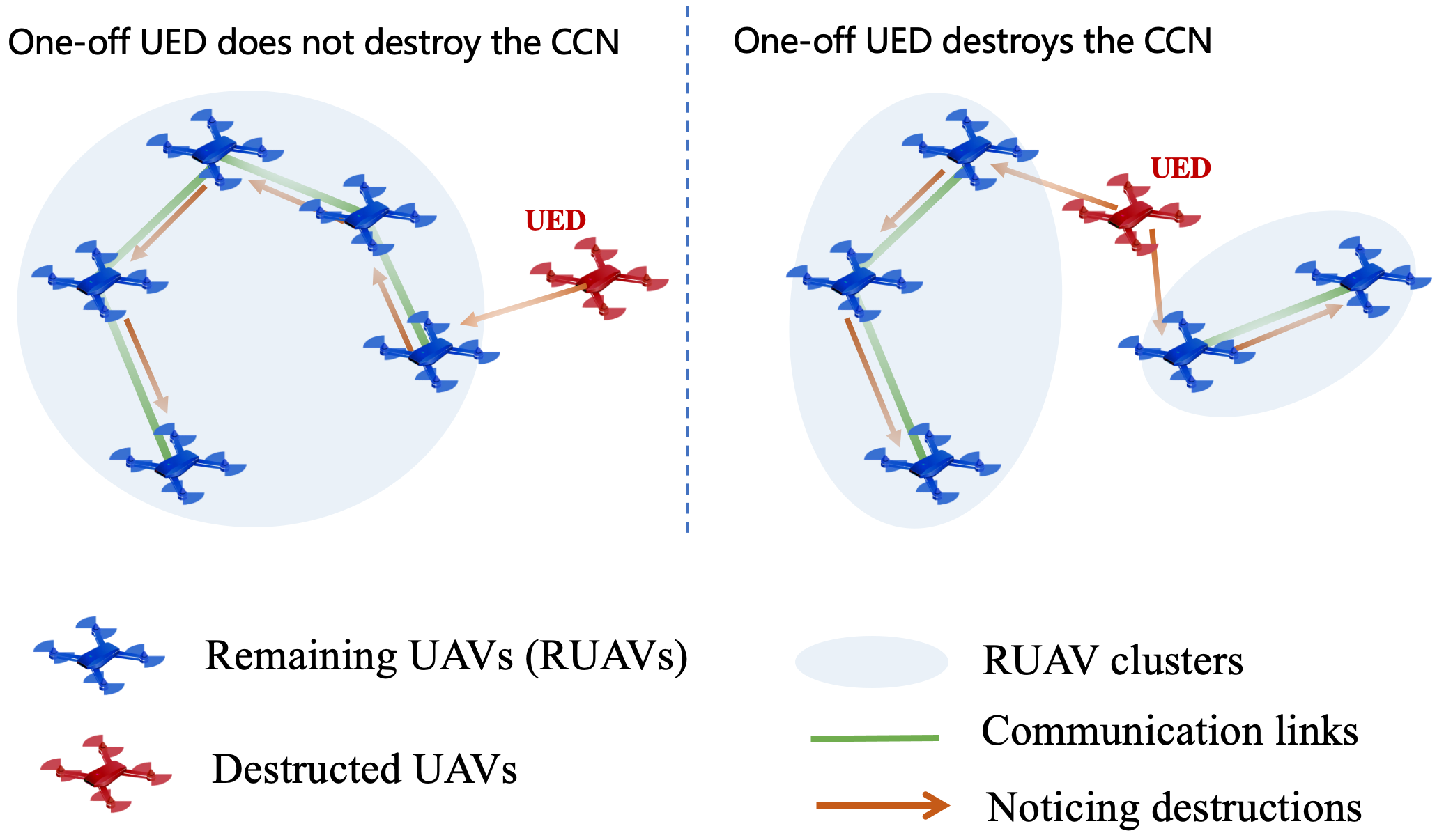}
	\caption{{Examples of one-off UEDs to the same USNET. The left one-off UED does not destroy the CCN, while the right one-off UED destroys the CCN.} }
	\label{fig:appendix_destroy_CCN}
\end{figure}
\begin{appendices}
\section{{Illustrations of  One-off UEDs cases}}
\label{appendix:type-one-off-UEDs}
{
Consider a USNET composed of $N$ UAVs with fixed initial positions $\{\mathbf{p}_{1,0},\mathbf{p}_{2,0},...,\mathbf{p}_{N,0}\}$. The one-off UED can destruct any number of UAVs with random indexes in the USNET at a certain time step. Denote the number of destructed UAVs as $\Upsilon\in\{1,2,...,N\}$. The number of {case}s of destructing $\Upsilon$ UAVs can be calculated as $C^{\Upsilon}_N=\frac{N!}{\Upsilon!(N-\Upsilon)!}$, where $C$ is the combinatorial number. Hence, the total number of one-off UED {case}s is  $\sum_{\Upsilon=1}^NC_N^\Upsilon=2^N$. }

{Note that not all {case}s of one-off UEDs can destroy the communication connectivity of the USNET. For example, as shown in \reffig{fig:appendix_destroy_CCN}, the one-off UED on the left does not destroy the CCN, while the one-off UED on the right destroys the CCN. The RUAVs can stay still if they remain a CCN after the one-off UED. Therefore, we only consider the one-off UEDs that can destroy the communication connectivity of the USNET. }

\section{{Proof of Proposition 1}}
\label{appendix:proposition_1}
{We first prove that the metric space $\{\mathbf{X}_t\mid \frac{1}{|\mathcal{I}_t|}\sum_{i\in\mathcal{I}_t}\mathbf{p}_{i,t}=\mathbf{c}\}$ is closed under the GCO $G(\cdot)$, i.e.,
\begin{align}
	G(\cdot):\{\mathbf{X}_t\mid \frac{1}{|\mathcal{I}_t|}\sum_{i\in\mathcal{I}_t}\mathbf{p}_{i,t}=\mathbf{c}\}\rightarrow\{\mathbf{X}_t\mid \frac{1}{|\mathcal{I}_t|}\sum_{i\in\mathcal{I}_t}\mathbf{p}_{i,t}=\mathbf{c}\}.
\end{align}
Then we prove that the GCO $G(\cdot)$ satisfies the contraction mapping theorem \cite{cm} when $0<H_t\le\frac{1}{\left\|\mathbf{A}^v_t\right\|_\infty}$. In addition, we prove that the positions of RUAVs in the topology matrix $\overline{\mathbf{X}}_t$ (Banach point \cite{cm}) of the GCO $G(\cdot)$ all have the same value $\mathbf{c}$, i.e., $\overline{\mathbf{X}}_t=[\mathbf{c},\mathbf{c},...,\mathbf{c}]^T$.}
\subsection{The Closure of GCO $G(\cdot)$ in $\{\mathbf{X}_t\mid \frac{1}{|\mathcal{I}_t|}\sum_{i\in\mathcal{I}_t}\mathbf{p}_{i,t}=\mathbf{c}\}$}
\label{appendix:proposition_1:closure}
{We need to prove that $\forall \mathbf{X}_t\in\{\mathbf{X}_t\mid \frac{1}{|\mathcal{I}_t|}\sum_{i\in\mathcal{I}_t}\mathbf{p}_{i,t}=\mathbf{c}\}$,  $\mathbf{X}^1_{t}= G(\mathbf{X}_t)\in\{\mathbf{X}_t\mid \frac{1}{|\mathcal{I}_t|}\sum_{i\in\mathcal{I}_t}\mathbf{p}_{i,t}=\mathbf{c}\}$ holds, i.e.,
\begin{align}
	\label{mean_position}
	\frac{1}{|\mathcal{I}_t|}\sum_{i\in\mathcal{I}_t}\mathbf{p}_{i,t}^1=\frac{1}{|\mathcal{I}_t|}\sum_{i\in\mathcal{I}_t}\mathbf{p}_{i,t}=\mathbf{c},
\end{align} 
where ${\mathbf{p}^1_{r_j,t}}^T$ is the $j$-th row of $\mathbf{X}^1_t$.}
Let $\mathbf{p}^1_{i,t}\triangleq[x^1_{i,t},y^1_{i,t},z^1_{i,t}]^T$, where $x^1_{i,t}$, $y^1_{i,t}$ and $z^1_{i,t}$ denote the $X$, $Y$ and $Z$ axis components of $\mathbf{p}^1_{i,t}$. Then \refeq{mean_position} is equivalent to
\begin{align}
	\label{center_proof}
	\sum_{i\in\mathcal{I}_t}x^1_{i,t}&=\sum_{i\in\mathcal{I}_t}x_{i,t},\;\text{and}\;
	\sum_{i\in\mathcal{I}_t}y^1_{i,t}=\sum_{i\in\mathcal{I}_t}y_{i,t},\notag\\
	&\;\text{and}\;
	\sum_{i\in\mathcal{I}_t}z^1_{i,t}=\sum_{i\in\mathcal{I}_t}z_{i,t}.
\end{align}
Let us prove $\sum_{i\in\mathcal{I}_t}x^1_{i,t}=\sum_{i\in\mathcal{I}_t}x_{i,t}$ in \eqref{center_proof}  as an example. 
Since $\mathbf{X}^1_t=(\mathbf{I}_t-H_t\mathbf{L}_t^v)\mathbf{X}_t$, we have \eqref{equ:expansion} as shown below in this page,
where $l^1_{jj'}$ is the element in the $j$-th row and the $j'$-th column of matrix $\mathbf{I}_t-K_t\mathbf{L}_t^v$, we have
\begin{align}
	\label{pro_2}
	\sum_{i\in\mathcal{I}_t}x^1_{i,t}=\sum_{j=1}^{|\mathcal{I}_t|}x^1_{r_j,t}=\sum_{j=1}^{|\mathcal{I}_t|}\sum_{j'=1}^{|\mathcal{I}_t|}l^1_{jj'}x_{r_{j'},t}=\sum_{j'=1}^{|\mathcal{I}_t|}x_{r_{j'},t}\bigg(\sum_{j=1}^{|\mathcal{I}_t|}l^1_{jj'}\bigg).
\end{align}
Since
\begin{align}
	\sum_{j=1}^{|\mathcal{I}_t|}l^1_{jj'}=1+H_td_{j,t}-H_t\sum_{j'=1}^{|\mathcal{I}_t|}a_{jj',t}=1,
\end{align}
we have 
\begin{align}
	\sum_{i\in\mathcal{I}_t}x^1_{i,t}=\sum_{i\in\mathcal{I}_t}x_{i,t}\bigg(\sum_{j=1}^{|\mathcal{I}_t|}l^1_{jj'}\bigg)=\sum_{i\in\mathcal{I}_t}x_{i,t}.
\end{align}

\stripsep=0pt
\begin{strip}
	\hrule
	\begin{align}
		\label{equ:expansion}
		\setlength{\arraycolsep}{1.2pt}
		\renewcommand{\arraystretch}{1.4}
		\begin{bmatrix}x_{r_1,t}^1&y_{r_1,t}^1&z_{r_1,t}^1\\x_{r_2,t}^1&y_{r_2,t}^1&z_{r_2,t}^1\\\vdots&\vdots&\vdots\\x^1_{r_{|\mathcal{I}_t|},t}&y^k_{r_{|\mathcal{I}_t|},t}&z^k_{r_{|\mathcal{I}_t|},t}\end{bmatrix}=\begin{bmatrix}l^1_{11}&l^1_{12}&\cdots&l^1_{1{|\mathcal{I}_t|}}\\l^1_{21}&l^1_{22}&\cdots&l^1_{2{|\mathcal{I}_t|}}\\\vdots&\vdots&\ddots&\vdots\\l^1_{{|\mathcal{I}_t|}1}&l^1_{|\mathcal{I}_t|2}&\cdots&l^1_{{|\mathcal{I}_t|}{|\mathcal{I}_t|}}\end{bmatrix}\begin{bmatrix}x_{r_1,t}&y_{r_1,t}&z_{r_1,t}\\x_{r_2,t}&y_{r_2,t}&z_{r_2,t}\\\vdots&\vdots&\vdots\\x_{r_{|\mathcal{I}_t|},t}&y_{r_{|\mathcal{I}_t|},t}&z_{r_{|\mathcal{I}_t|},t}\end{bmatrix},
	\end{align}
\end{strip}
The equalities  $\sum_{i\in\mathcal{I}_t}y^1_{i,t}=\sum_{i\in\mathcal{I}_t}y_{i,t}$ and $\sum_{i\in\mathcal{I}_t}z^1_{i,t}=\sum_{i\in\mathcal{I}_t}z_{i,t}$ can be proved in the same manner. Therefore, \eqref{mean_position} holds.

\subsection{Satisfaction of Contraction Mapping Theorem}

In the metric space $\{\mathbf{X}_t\mid \frac{1}{|\mathcal{I}_t|}\sum_{i\in\mathcal{I}_t}\mathbf{p}_{i,t}=\mathbf{c}\}$, {we define} the distance between {any two topology matrices} $\mathbf{X}_t'$ and $\mathbf{X}_t''$ as 
\begin{align}
	d(\mathbf{X}_t',\mathbf{X}_t'')&=\left\|\mathbf{X}_t'-\mathbf{X}_t''\right\|_\infty\notag\\
	&=\max_{j\in\{1,...,|\mathcal{I}_t|\}}\{\sum_{s=1}^{3}|(\mathbf{X}_t'-\mathbf{X}_t'')_{js}|\}.
\end{align}
The distance between the GCO $G(\cdot)$ of $\mathbf{X}_t'$ and $\mathbf{X}_t''$ can be calculated as 
\begin{align}
	d(G(\mathbf{X}_t'),G(\mathbf{X}_t''))&=\left\|G(\mathbf{X}_t')-G(\mathbf{X}_t'')\right\|_\infty
	\notag\\
	&=\left\|(\mathbf{I}_t-H_t\mathbf{L}_t^v)(\mathbf{X}_t'-\mathbf{X}_t'')\right\|_\infty.
\end{align}
{Since the matrix infinity norm $\left\|\cdot\right\|_\infty$ has the sub-multiplicity property \footnote{{We prove the sub-multiplicity of $\left\|\cdot\right\|_\infty$ in Appendix \ref{appendix:submultiplicity}.}}, we have }
\begin{align}
	\label{equ:submul_contraction}
	\left\|(\mathbf{I}_t-H_t\mathbf{L}_t^v)(\mathbf{X}_t'-\mathbf{X}_t'')\right\|_\infty\le\left\|\mathbf{I}_t-H_t\mathbf{L}_t^v\right\|_\infty\left\|\mathbf{X}_t'-\mathbf{X}_t''\right\|_\infty,
\end{align}
{Thus, we can get}
\begin{align}
	&\quad\;\; d(G(\mathbf{X}_t'),G(\mathbf{X}_t''))\le\left\|\mathbf{I}_t-H_t\mathbf{L}_t^v\right\|_\infty\left\|\mathbf{X}_t'-\mathbf{X}_t''\right\|_\infty\notag\\
	&=\left\|\mathbf{I}_t-H_t(\mathbf{D}_t^v-\mathbf{A}_t^v)\right\|_\infty\left\|\mathbf{X}_t'-\mathbf{X}_t''\right\|_\infty\notag\\
	&=\max_{i\in\mathcal{I}_t}\bigg[|1-H_td^v_{i,t}|+\sum_{i'\in\mathcal{I}_t}|H_ta^v_{ii',t}|\bigg]\left\|\mathbf{X}_t'-\mathbf{X}_t''\right\|_\infty\notag\\
	&=\max_{i\in\mathcal{I}_t}\bigg[|1-H_t\sum_{i'\in\mathcal{I}_t}a^v_{ii',t}|+\sum_{i'\in\mathcal{I}_t}H_ta^v_{ii',t}\bigg]\left\|\mathbf{X}_t'-\mathbf{X}_t''\right\|_\infty.
\end{align}
When $H_t\le\frac{1}{\left\|\mathbf{A}^v_t\right\|_\infty}$, there is 
\begin{align}
	\label{K_2}
	1-H_t\sum_{i'\in\mathcal{I}_t}a^v_{ii',t}&\ge 1-\frac{1}{\left\|\mathbf{A}^v_t\right\|_\infty}\sum_{i'\in\mathcal{I}_t}a^v_{ii',t}\notag\\&\ge 1-\frac{1}{\left\|\mathbf{A}^v_t\right\|_\infty}\left\|\mathbf{A}^v_t\right\|_\infty=0,
\end{align}
and we have
\begin{align}
	\label{equ:contraction}
	&\quad\;\; d(G(\mathbf{X}_t'),G(\mathbf{X}_t''))\notag\\
	&\le\max_{i\in\mathcal{I}_t}\bigg[|1-H_t\sum_{i'\in\mathcal{I}_t}a^v_{ii',t}|+\sum_{i'\in\mathcal{I}_t}H_ta^v_{ii',t}\bigg]\left\|\mathbf{X}_t'-\mathbf{X}_t''\right\|_\infty\notag\\
	&=\max_{i\in\mathcal{I}_t}\bigg[1-H_t\sum_{i'\in\mathcal{I}_t}a^v_{ii',t}+H_t\sum_{i'\in\mathcal{I}_t}a^v_{ii',t}\bigg]\left\|\mathbf{X}_t'-\mathbf{X}_t''\right\|_\infty\notag\\
	&=\max_{i\in\mathcal{I}_t}[1]\left\|\mathbf{X}_t'-\mathbf{X}_t''\right\|_\infty\notag\\
	&=d(\mathbf{X}_t',\mathbf{X}_t'').
\end{align} 
{The condition for \eqref{equ:contraction} to be equal is that \eqref{equ:submul_contraction} takes the equal sign, i.e., 
\begin{align}
	\label{equ:equal_}
		\left\|(\mathbf{I}_t-H_t\mathbf{L}_t^v)(\mathbf{X}_t'-\mathbf{X}_t'')\right\|_\infty=\left\|\mathbf{I}_t-H_t\mathbf{L}_t^v\right\|_\infty\left\|\mathbf{X}_t'-\mathbf{X}_t''\right\|_\infty.
\end{align}
As shown in Appendix \ref{appendix:submultiplicity}, when \eqref{equ:equal_} holds, we can draw two inferences:
\begin{enumerate}
\item \emph{inference 1:} $\forall j'\in\{1,2,...,|\mathcal{I}_t|\},s\in\{1,2,3\}$, when $j=\arg\max_{j}\sum_{s=1}^3\left|\sum_{j'=1}^{|\mathcal{I}_t|}(\mathbf{I}_t-H_t\mathbf{L}_t^v)_{jj'}(\mathbf{X}_t'-\mathbf{X}_t'')_{j's}\right|$, we have $(\mathbf{I}_t-H_t\mathbf{L}_t^v)_{jj'}(\mathbf{X}_t'-\mathbf{X}_t'')_{j's}\ge 0$;
\item \emph{inference 2:}  $\sum_{s=1}^3|(\mathbf{X}_t'-\mathbf{X}_t'')_{j's}|=C',\;\forall j'\in\{1,2,...,|\mathcal{I}_t|\}$, where $C'\in\mathbb{R}$ is a constant.
\end{enumerate} 
When $H_t\le\frac{1}{\left\|\mathbf{A}^v_t\right\|_\infty}$, each element in $\mathbf{I}_t-H_t\mathbf{L}_t^v$ is not smaller than $0$, i.e., $(\mathbf{I}_t-H_t\mathbf{L}_t^v)_{jj'}\ge 0,\forall j,j'$. Hence, from \emph{inference 1}, we can derive $(\mathbf{X}_t'-\mathbf{X}_t'')_{j's}\ge 0, \forall j',s$. With \emph{inference 2}, we have
\begin{align}
	\sum_{j'}^{|\mathcal{I}_t|}\sum_{s=1}^3|(\mathbf{X}_t'-\mathbf{X}_t'')_{j's}|=\sum_{j'}^{|\mathcal{I}_t|}\sum_{s=1}^3(\mathbf{X}_t'-\mathbf{X}_t'')_{j's}=|\mathcal{I}_t|C'.
\end{align} 
Since $\mathbf{X}_t',\mathbf{X}_t''\in\{\mathbf{X}_t\mid \frac{1}{|\mathcal{I}_t|}\sum_{i\in\mathcal{I}_t}\mathbf{p}_{i,t}=\mathbf{c}\}$, we can derive
\begin{align}
C'&=\frac{1}{|\mathcal{I}_t|}\sum_{s=1}^3\bigg[\sum_{j'}^{|\mathcal{I}_t|}(\mathbf{X}_t')_{j's}-(\mathbf{X}_t'')_{j's}\bigg]\notag\\
&=\frac{1}{|\mathcal{I}_t|}|\mathcal{I}_t|(\text{sum}(\mathbf{c})-\text{sum}(\mathbf{c}))=0,
\end{align}
where $\text{sum}(\cdot)$ represents the summation of all the elements in vectors.
This indicates that $(\mathbf{X}_t'-\mathbf{X}_t'')_{j's}=0,\forall j',s$. Hence, when \eqref{equ:contraction} takes the equal sign, we have
\begin{align}
	d(G(\mathbf{X}_t'),G(\mathbf{X}_t''))=d(\mathbf{X}_t',\mathbf{X}_t'')=\left\|\mathbf{X}_t'-\mathbf{X}_t''\right\|_\infty=0.
\end{align}
Thereby, we have proved that $\forall \mathbf{X}_t',\mathbf{X}_t''\in\{\mathbf{X}_t\mid \frac{1}{|\mathcal{I}_t|}\sum_{i\in\mathcal{I}_t}\mathbf{p}_{i,t}=\mathbf{c}\}$,
\begin{align}
	d(G(\mathbf{X}_t'),G(\mathbf{X}_t''))\le\delta d(\mathbf{X}_t',\mathbf{X}_t''),
\end{align}
where $\delta\in(0,1)$.}
Hence, the GCO $G(\cdot)$ is a {contraction} mapping when $0<H_t\le\frac{1}{\left\|\mathbf{A}^v_t\right\|_\infty}$. There exists and only exists one topology matrix $\overline{\mathbf{X}}_t$ {(the Banach point of the GCO $G(\cdot)$)} such that 
\begin{align}
	\label{contraction}
	\overline{\mathbf{X}}_t=G(\overline{\mathbf{X}}_t)=\lim_{k\rightarrow\infty}G^k(\mathbf{X}_t).
\end{align}
\subsection{Property of $\overline{\mathbf{X}}_t=[\mathbf{c},\mathbf{c},...,\mathbf{c}]^T$}
Since $	\overline{\mathbf{X}}_t=G(\overline{\mathbf{X}}_t)$, we have
\begin{align}
\overline{\mathbf{X}}_t=(\mathbf{I}_t-H_t\mathbf{L}_t^v)\overline{\mathbf{X}}_t.
\end{align}
Eliminating $\overline{\mathbf{X}}_t$ on both sides of \refeq{contraction}, we have
\begin{align}
	-H_t\mathbf{L}_t^v\overline{\mathbf{X}}_t=0\;
	\Rightarrow\;\mathbf{L}_t^v\overline{\mathbf{X}}_t=0\;\Rightarrow\;\mathbf{L}_t^v[\overline{\mathbf{x}}_{1,t},\overline{\mathbf{x}}_{2,t},\overline{\mathbf{x}}_{3,t}]=0,
\end{align}
where $\overline{\mathbf{x}}_{s,t},s\in\{1,2,3\}$ is the $s$-th column vector of $\overline{\mathbf{X}}_t$. Furthermore, $\overline{\mathbf{x}}_{s,t}$ is the eigenvector of $\mathbf{L}^v_t$ corresponding to zero eigenvalue, since $\mathbf{L}^v_t\overline{\mathbf{x}}_{s,t}=0=0\overline{\mathbf{x}}_{s,t}$. Note that the VRG is a CCN, and the algebraic multiplicity of the zero  eigenvalue of $\mathbf{L}^v_t$ 
equals to 1. Hence, the eigenvectors can only be the multiple of $\mathbf{1}_{|\mathcal{R}_t|}$, i.e., $\overline{\mathbf{x}}_{s,t}=\alpha_{s}\mathbf{1}_{|\mathcal{I}_t|}$, where $\alpha_{s}\in\mathbb{R}$ {is a constant}, and $\alpha_{s}\ne 0$. Then we have
\begin{align}
	\label{gather_same}
	\overline{\mathbf{X}}_t=[\alpha_{1}\mathbf{1}_{|\mathcal{I}_t|}, \alpha_{2}\mathbf{1}_{|\mathcal{I}_t|},\alpha_{3}\mathbf{1}_{|\mathcal{I}_t|}]=[\overline{\mathbf{p}}_{r_1,t},\overline{\mathbf{p}}_{r_2,t},...,\overline{\mathbf{p}}_{r_{|\mathcal{I}_t|},t}]^T,
\end{align}
where $\overline{\mathbf{p}}_{r_j,t}=[\alpha_{1},\alpha_{2},\alpha_{3}]^T$. Equation \refeq{gather_same} indicates that iteratively applying the GCO $G(\cdot)$ to the $\mathbf{X}_t$ will gather all RUAVs to a same position $[\alpha_{1},\alpha_{2},\alpha_{3}]^T$. Since $\overline{\mathbf{X}}_t\in\{\mathbf{X}_t\mid \frac{1}{|\mathcal{I}_t|}\sum_{i\in\mathcal{I}_t}\mathbf{p}_{i,t}=\mathbf{c}\}$, we have 
\begin{align}
	\frac{1}{|\mathcal{I}_t|}\sum_{i\in\mathcal{I}_t}\overline{\mathbf{p}}_{i,t}=\frac{1}{|\mathcal{I}_t|}\sum_{i\in\mathcal{I}_t}[\alpha_{1},\alpha_{2},\alpha_{3}]^T=[\alpha_{1},\alpha_{2},\alpha_{3}]^T=\mathbf{c}.
\end{align}
Hence, we have $\overline{\mathbf{X}}_t=[\overline{\mathbf{p}}_{r_1,t},\overline{\mathbf{p}}_{r_2,t},...,\overline{\mathbf{p}}_{r_{|\mathcal{I}_t|},t}]^T=[\mathbf{c},\mathbf{c},...,\mathbf{c}]^T$.

\section{{Proof of Proposition 2}}
\label{appendix:proposition_2}
Consider the GCO $G(\cdot)$ in metric space $\{\mathbf{X}_t\mid \frac{1}{|\mathcal{I}_t|}\sum_{i\in\mathcal{I}_t}\mathbf{p}_{i,t}=\mathbf{c}\}$, where $\mathbf{c}$ is the center of all RUAVs.
From Appendix \ref{appendix:proposition_1:closure}, we know that $\mathbf{X}^k_{t}\in\{\mathbf{X}_t\mid \frac{1}{|\mathcal{I}_t|}\sum_{i\in\mathcal{I}_t}\mathbf{p}_{i,t}=\mathbf{c}\},\forall k\in\mathbb{N}_+$. As the GCO $G(\cdot)$ is a contraction mapping, we have
\begin{align}
	d(\mathbf{X}^{k+1}_{t},\overline{\mathbf{X}}_t)=d(G(\mathbf{X}^{k}_{t}),G(\overline{\mathbf{X}}_t))\le\delta d(\mathbf{X}^{k}_{t},\overline{\mathbf{X}}_t),\;\forall k\in\mathbb{N}_+,
\end{align}
which means
\begin{align}
	\max_{i\in\mathcal{I}_t}\left\|\mathbf{p}^{k+1}_{i,t}-\mathbf{c}\right\|_1\le\delta	\max_{i\in\mathcal{I}_t}\left\|\mathbf{p}^{k}_{i,t}-\mathbf{c}\right\|_1,
\end{align}
where $\delta\in(0,1)$, and $\left\|\cdot\right\|_1$ represents the 1-norm operator of vectors. Hence, the positions of RUAVs are moving towards the center of their positions $\frac{1}{|\mathcal{I}_t|}\sum_{i\in\mathcal{I}_t}\mathbf{p}_{i,t}=\mathbf{c}$.
\section{{Proof of the sub-multiplicity of $\left\|\cdot\right\|_\infty$}}
\label{appendix:submultiplicity}
Consider two arbitrary matrices $\mathbf{A}=(a_{ij})\in\mathbb{R}^{m\times n}$ and $\mathbf{B}=(b_{jk})\in\mathbb{R}^{n\times r}$, where $m,n,r\in\mathbb{R}$. We have
\begin{align}
	\label{equ:submul}
	\left\|\mathbf{A}\mathbf{B}\right\|_\infty&=\max_{i\in\{1,...,m\}}\sum_{j=1}^r\left|\sum_{k=1}^na_{ik}b_{kj}\right|\notag\\
	&\le \max_{i\in\{1,...,m\}}\sum_{j=1}^r\sum_{k=1}^n|a_{ik}||b_{kj}|\notag\\
	&=\max_{i\in\{1,...,m\}}\sum_{k=1}^n|a_{ik}|\bigg(\sum_{j=1}^r|b_{kj}|\bigg)\notag\\
	&\le \max_{i\in\{1,...,m\}}\sum_{k=1}^n|a_{ik}|\bigg(\max_{k\in\{1,...,n\}}\sum_{j=1}^r|b_{kj}|\bigg)\notag\\
	&=\left\|\mathbf{A}\right\|_\infty\left\|\mathbf{B}\right\|_\infty.
\end{align}
Hence, the sub-multiplicity of $\left\|\cdot\right\|_\infty$ holds. The equality condition for \eqref{equ:submul} is that
\begin{enumerate}
	\item for $i=\arg\max_{i\in\{1,...,m\}}\sum_{j=1}^r\left|\sum_{k=1}^na_{ik}b_{kj}\right|$, $a_{ik}b_{kj}\ge 0,\; \forall k\in\{1,2,...,n\},j\in\{1,2,...,r\}$,
	\item $\sum_{j=1}^r|b_{kj}|=C',\;\forall k\in\{1,2,...,n\}$, where $C'\in\mathbb{R}$ is a constant
\end{enumerate}
hold at the same time.

\end{appendices}



\end{document}